\documentclass[11pt]{amsart}
\allowdisplaybreaks
\usepackage{amssymb,mathrsfs,graphicx,enumerate}
\usepackage{amsmath,amsfonts,amssymb,amscd,amsthm,bbm}
\usepackage[retainorgcmds]{IEEEtrantools}
\usepackage{tikz-cd}
\usepackage{makecell}
\usepackage{tikz}
\usepackage{array}
\usetikzlibrary{matrix}
\usepackage{pgfpages}
\usepackage{colortbl}

\usepackage{tabto}
\usepackage{pgfplots}
\usepackage{asypictureB}
\usepackage{graphicx, subfigure}

\title[A second-order Lohe hermitian model]{Collective behaviors of the Lohe hermitian sphere model with inertia}

\author[Ha]{Seung-Yeal Ha}
\address[Seung-Yeal Ha]{\newline Department of Mathematical Sciences\newline Seoul National University, Seoul 08826 and \newline
Korea Institute for Advanced Study, Hoegiro 85, 02455, Seoul, Republic of Korea}
\email{syha@snu.ac.kr}

\author[Kang]{Myeongju Kang}
\address[Myeongju Kang]{\newline Department of Mathematical Sciences\newline Seoul National University, Seoul 08826, Republic of Korea}
\email{bear0117@snu.ac.kr}

\author[Park]{Hansol Park}
\address[Hansol Park]{\newline Department of Mathematical Sciences\newline Seoul National University, Seoul 08826, Republic of Korea}
\email{hansol960612@snu.ac.kr}

\newtheorem{theorem}{Theorem}[section]
\newtheorem{lemma}{Lemma}[section]
\newtheorem{corollary}{Corollary}[section]
\newtheorem{proposition}{Proposition}[section]
\newtheorem{remark}{Remark}[section]

\newtheorem{definition}{Definition}[section]

\newcommand{\bbc}{\mathbb C}
\newcommand{\bbh}{\mathbb H}

\newcommand{\bbr}{\mathbb R}
\newcommand{\bbs}{\mathbb S}

\newcommand{\cI}{\mathcal I}
\newcommand{\cT}{\mathcal T}

\begin{document}

\date{\today}


\thanks{\textbf{Acknowledgment.} The work of S.-Y.Ha is supported by NRF-2020R1A2C3A01003881, the work of M. Kang was supported by the National Research Foundation of Korea(NRF) grant funded by the Korea government(MSIP)(2016K2A9A2A13003815), and the work of H. Park was supported by Basic Science Research Program through the National Research Foundation of Korea(NRF) funded by the Ministry of Education (2019R1I1A1A01059585)}

\begin{abstract}
We present a second-order extension of the first-order Lohe hermitian sphere(LHS) model and study its emergent asymptotic dynamics. Our proposed model incorporates  an inertial effect as a second-order extension. The inertia term can generate an oscillatory behavior of particle trajectory in a small time interval(initial layer) which causes a technical difficulty for the application of monotonicity-based arguments. For emergent estimates, we employ two-point correlation function which is defined as an inner product between positions of particles. For a homogeneous ensemble with the same frequency matrix, we provide two sufficient frameworks in terms of system parameters and initial data to show that two-point correlation functions tend to the unity which is exactly the same as the complete aggregation. In contrast, for a heterogeneous ensemble with distinct frequency matrices, we provide a sufficient framework in terms of system parameters and initial data, which makes two-point correlation functions close to unity by increasing the principal coupling strength.\end{abstract}

\dedicatory{Dedicated to the celebration of the 80th birthday of Prof. Shuxing Chen}

\maketitle \centerline{\date}

%\tableofcontents

%%%%%%%%%%%%%%%%%%%%%%%%%%%%%%%%%%
%
%  Section 1
%
%%%%%%%%%%%%%%%%%%%%%%%%%%%%%%%%%%
\section{Introduction}
\setcounter{equation}{0}
Collective behaviors of a many-body system are often observed in biological complex networks, to name a few, flocking of birds, swarming of fish, herding of sheep, synchronous firing of fireflies, neurons and pacemaker cells \cite{A-B, B-B, C-U, H-K-P-Z, Ku1, Ku2, Pe, P-R-K, St, Wi2} etc. In this paper, we are interested in the aggregation phenomena of particles on a Hermitian sphere. To motivate our discussion, we begin with the first-order LHS model.

Let $z_j = z_j(t)$ be the position of the $j$-th particle on a Hermitian sphere at time $t$. In order to fix the idea, we begin with the first-order Lohe hermitian sphere model \cite{H-P0, H-P1, H-P2} on a Hermitian sphere ${\mathbb H}{\mathbb S}_r^d := \{ z \in {\mathbb C}^{d+1}:~ \|z \| = r \}$:
\begin{equation} \label{A-0}
{\dot z}_j = \Omega_j z_j + \kappa_0 ( \langle z_j, z_j \rangle z_c-\langle{z_c, z_j}\rangle z_j)+\kappa_1(\langle{z_j, z_c}\rangle-\langle{z_c, z_j}\rangle)z_j, \quad j = 1, \cdots, N,
\end{equation}
where $\langle z_1, z_2 \rangle := \sum_{\alpha=1}^{d+1} \overline{z_1^\alpha} z_2^\alpha$ which is conjugate linear in the first argument and linear in the second argument, and where $\Omega_j$ is the skew-symmetric $(d+1) \times (d+1)$ matrix such that $\Omega_j^{\dagger} = -\Omega_j$ and $\kappa_0, \kappa_1$ are nonnegative coupling strengths.  The emergent dynamics of the HLS model \eqref{A-0} has been extensively studied in \cite{H-P0, H-P1, H-P2} (see Section \ref{sec:2.1}). 

In this paper, we are interested in the large-time dynamics of the Cauchy problem to the second-order extension of \eqref{A-1} incorporating inertial effect to \eqref{A-0}:
\begin{equation}\label{A-1}
\begin{cases}
\displaystyle m \Big( \dot{v}_j -\frac{\Omega_j}{\gamma} v_j \Big) +\gamma v_j = \kappa_0(\langle{z_j, z_j}\rangle z_c-\langle{z_c, z_j}\rangle z_j) +\kappa_1(\langle{z_j, z_c}\rangle -\langle{z_c, z_j}\rangle)z_j\\
\displaystyle \hspace{2.5cm} -\frac{m \| v_j \|^2}{\| z_j \|^2}z_j, \quad z_j \in {\mathbb H}{\mathbb S}_r^d,~~t > 0,  \\
\displaystyle v_j = \dot{z}_j -\frac{\Omega_j}{\gamma} z_j, \quad z_j(0)=z_j^{in}, \quad v_j(0)=v_j^{in}, \quad \langle{z_j^{in}, v_j^{in}}\rangle+\langle{v_j^{in}, z_j^{in}}\rangle = 0,
\end{cases}
\end{equation}
where $m$ is the strength of inertia which is nonnegative. \newline

It is easy to see that for zero inertia and unit friction constant, system \eqref{A-1} reduces to the LHS model \eqref{A-0}. The basic conservation law and solution splitting property will be given in Lemma \ref{L2.1} and Lemma \ref{L2.2}, respectively. Before we present our main results, we first recall the concepts of complete aggregation and practical aggregation as follows. 
\begin{definition} \label{D1.1}
Let $Z := \{ z_j \}$ be a solution to \eqref{A-1}.
\begin{enumerate}
\item
The solution $Z$ exhibits (asymptotic) complete aggregation if the following estimate holds.
\[ \lim_{t \to \infty} \max_{i, j} \|z_i(t) - z_j(t) \| = 0. \]
\item
The solution $Z$ exhibits (asymptotic) practical aggregation if the following estimate holds.
\[ \lim_{\kappa_0 \to \infty} \limsup_{0 \leq t < \infty} \max_{i, j} \|z_i(t) - z_j(t) \| = 0. \]
\end{enumerate}
\end{definition}

\vspace{0.2cm}

Next, we briefly present our two main results on the large-time emergent dynamics of \eqref{A-1}. \newline

First, we present a sufficient framework for the complete aggregation for the homogeneous ensemble with $\Omega_j = \Omega$. In this case, we may
assume that $\Omega = 0$ and $z_j$ satisfies 
\[ m \ddot{z}_j = -\gamma \dot{z}_j + \kappa_0 \big(z_c -\left\langle{z_c, z_j }\right\rangle z_j \big) +\kappa_1 \big( \left\langle{ z_j, z_c }\right\rangle -\left\langle{ z_c, z_j }\right\rangle \big) z_j - m \| \dot{z}_j \|^2 z_j.  \]
Under the following conditions on system parameters and initial data:
\[ \gamma \gg m, \quad  {\mathcal G}(0) \ll 1, \quad  |\dot{{\mathcal G}}(0)| + {\mathcal G}(0)  \ll 1. \]
For the detailed conditions, we refer to frameworks $({\mathcal F}_A1)$-$({\mathcal F}_A2)$ and $({\mathcal F}_B1)$-$({\mathcal F}_B2)$ in Section \ref{sec:4.1}. Our first main result is concerned with the complete aggregation (see Theorem \ref{T4.1}):
\[ \lim_{t \to \infty} \max_{i,j} |z_i(t) - z_j(t) |= 0. \]
For this, we introduce two-point correlation functions:
\[  h_{ij} = \langle{ z_i, z_j \rangle}, \quad g_{ij} := 1 -h_{ij}, \quad 1 \leq i,j \leq N, \quad {\mathcal G} := \frac{1}{N^2}\sum_{i, j=1}^N | g_{ij} |^2, \]
and  then, we also derive differential inequality for ${\mathcal G}$:
\[ m \ddot{{\mathcal G}} +\gamma \dot{{\mathcal G}} +4\kappa_0 \delta {\mathcal G}  \leq f(t), \quad f(t)\to0 \quad \mbox{as $t \to \infty$}. \]
Then, via Gronwall's differential inequality, we can derive the zero convergence of ${\mathcal G}$:
\[ \lim_{t \to \infty} {\mathcal G}(t) = 0, \quad \mbox{i.e.,} \quad \lim_{t \to \infty} \langle z_i(t), z_j(t) \rangle = 1, \quad \forall~i,j = 1, \cdots, N. \]
This clearly implies the complete aggregation in the sense of Definition \ref{D1.1}. \newline

Second, we deal with a heterogeneous ensemble with distinct natural frequency matrices $\Omega_i$. In this situation, we derive a rather weak aggregation, namely practical aggregation. For this, we propose a framework on the system parameters and initial data:
\[ \gamma \gg m, \quad  {\mathcal G}(0) \ll 1, \quad  |\dot{{\mathcal G}}(0)| + {\mathcal G}(0)  \ll 1. \]
For the detailed conditions, we refer to frameworks $({\mathcal F}_C1)$-$({\mathcal F}_C2)$ in Section \ref{sec:4.2}.  As in the aggregation estimate to the homogeneous ensemble, we derive a second-order Gronwall's inequality:
\[
m\ddot{{\mathcal G}} +\gamma \dot{{\mathcal G}} +4\kappa_0 \delta {\mathcal G} \leq 4\Omega^\infty +8\kappa_1 +\frac{16m }{\gamma^2}\big[ \Omega^\infty +2(\kappa_0 +\kappa_1) \big]^2, \quad t > 0.
\]
Then, via the second-order Gronwall's lemma (Lemma \ref{L6.2}) and a suitable ansatz for $m = \frac{m_0}{\kappa_0^{1+\eta}}$, one can show
\[ {\mathcal G}(t) \lesssim \max \Big \{ \frac{1}{\kappa_0},~\frac{1}{\kappa_0^{\eta}} \Big \}, \quad \mbox{for $t \gg 1$}. \]
This clearly implies the practical aggregation in Definition \ref{D1.1}. We refer to Theorem \ref{T4.2} for a detailed discussion. \newline

The rest of this paper is organized as follows. In Section \ref{sec:2}, we briefly introduce the second-order LHS model and basic properties of the proposed model and  discuss it with other previous models such as the first-order LHS model and the Kuramoto model, and review the previous result on the first-order Lohe Hermitian model. In Section \ref{sec:3},  we study the characterization and instability of some distinguished states. In Section \ref{sec:4}, we summarize our main results on the emergent dynamics of the second-order LHS model. In Section \ref{sec:5} and Section \ref{sec:6}, we provide proofs of Theorem \ref{T4.1} and Theorem \ref{T4.2}. Finally, Section \ref{sec:7} is devoted to a brief summary of our main results and discussion on some remaining problems for a future work. 

\vspace{0.3cm}

\noindent\textbf{Notation}: 
For a vector $z=(z^1,\cdots,z^{d+1}) \in \bbc^{d+1}$ and $w = (w^1, \cdots, w^{d+1}) \in {\mathbb C}^{d+1}$, we set the inner product $\langle \cdot, \cdot \rangle$ and its corresponding $\ell^2$-nrom:
\begin{equation*}
\langle z, w \rangle := \sum_{i=1}^{d +1} \overline{z^i} w^i, \quad \|z\| := \sqrt{\langle z, z \rangle}, \quad {\mathbb H}{\mathbb S}^d = {\mathbb H}{\mathbb S}_1^d.
\end{equation*}
For a given configuration $C := \{(z_j, w_j := \dot{z}_j) \}$, we set state and velocity diameters as follows.
\[ {\mathcal D}(Z) := \max_{i,j} |z_i - z_j |, \quad {\mathcal D}(W) := \max_{i,j} |w_i - w_j|. \]
%%%%%%%%%%%%%%%%%%%%%%%%%%%%%%%%%
%
%  Section 2.
%
%%%%%%%%%%%%%%%%%%%%%%%%%%%%%%%%%
\section{Preliminaries} \label{sec:2}
\setcounter{equation}{0}
In this section, we briefly  introduce a second-order LHS model \eqref{A-1} and its basic properties, and discuss its relations with other aggregation models such as the first-order LHS model and the Kuramoto model. 

\subsection{The second-order LHS model} \label{sec:2.1} In this subsection, we study basic properties of the seocnd-order LHS model. To factor out the rotational motion, we introduce an auxiliary variable $u_j$ on ${\mathbb H}{\mathbb S}_R^d$:
\begin{equation} \label{B-2}
z_j :=e^{\frac{t}{\gamma} \Omega_j} u_j, \quad j = 1, \cdots, N. 
\end{equation}
By direct calculations, one has 
\begin{equation} \label{B-3}
u_j = e^{-\frac{t}{\gamma} \Omega_j } z_j, \qquad \dot{u}_j = e^{-\frac{t}{\gamma} \Omega_j} v_j, \qquad \ddot{u}_j = e^{-\frac{t}{\gamma} \Omega_j } \Big( \dot{v}_j -\frac{1}{\gamma} \Omega_j v_j \Big), \qquad j = 1, \cdots, N.
\end{equation}
We substitute \eqref{B-3} into \eqref{A-1} and use the fact that $\Omega_j$ is skew-Hermitian to derive the equations for $u_j$:
\begin{align}
\begin{cases} \label{B-4}
\displaystyle m \ddot{u}_j +\gamma \dot{u}_j = \frac{\kappa_0}{N} \sum_{k=1}^N \bigg(  \| u_j \|^2 e^{\frac{\Omega_k -\Omega_j}{\gamma} t} u_k -\left\langle{ e^{\frac{\Omega_k -\Omega_j}{\gamma} t} u_k, u_j}\right\rangle u_j \bigg) \\
\displaystyle \hspace{2cm} +\frac{\kappa_1}{N} \sum_{k=1}^N \bigg( \left\langle {u_j, e^{\frac{\Omega_k -\Omega_j}{\gamma} t} u_k} \right\rangle -\left\langle { e^{\frac{\Omega_k -\Omega_j}{\gamma} t} u_k, u_j} \right\rangle \bigg) u_j -\frac{m \| \dot{u}_j \|^2}{\| u_j \|^2} u_j, \\
\displaystyle u_j(0) = u_j^{in} = z_j^{in}, \quad \dot{u}_j(0) = \dot{u}_j^{in} = v_j^{in}, \quad \langle{u_j^{in}, \dot{u}_j^{in}}\rangle+\langle{\dot{u}_j^{in}, u_j^{in}}\rangle = 0.
\end{cases}
\end{align}
In the following lemma, we study the conservation of $\|z_j \|$ and $\|u_j \|$.
\begin{lemma}\label{L2.1}
\emph{(Conservation laws)}
Let $\{ z_j \}$ and $\{ u_j \}$  be global solutions of \eqref{A-1} and \eqref{B-4}, respectively. Then, $\ell^2$-norms $\|z_j\|$ and $\| u_j \|$ are conserved quantity:
\[
\frac{d}{dt}\|z_j(t) \| = 0, \quad \frac{d}{dt}\|u_j(t)\| = 0, \quad \forall ~t\geq0, \quad j = 1,\cdots, N.
\]
\end{lemma}
\begin{proof} Since $e^{-\frac{t}{\gamma} \Omega_j}$ is unitary,  one can see 
\[  \|u_j \|^2 = \langle u_j, u_j \rangle = \Big\langle e^{-\frac{t}{\gamma} \Omega_j} z_j, e^{-\frac{t}{\gamma} \Omega_j} z_j \Big\rangle = \Big\langle e^{-\frac{t}{\gamma} \Omega_j}\Big(  e^{-\frac{t}{\gamma} \Omega_j} \Big)^\dagger z_j, z_j  \Big\rangle = \langle z_j, z_j \rangle = \|z_j \|^2. \]
Hence, we only verify the conservation of the norm $\|u_j \|$. Now we claim:
\begin{equation*} \label{B-5}
\frac{d}{dt}\|u_j \|^2 = \langle u_j, \dot{u}_j \rangle +\langle \dot{u}_j, u_j \rangle = 0.
\end{equation*}
Simple calculation yields
\begin{equation} \label{B-6}
m \frac{d}{dt} \Big (\langle u_j, \dot{u}_j \rangle +\langle \dot{u}_j, u_j \rangle \Big ) = 2m\|\dot{u}_j\|^2 +\langle u_j, m\ddot{u}_j \rangle +\langle m\ddot{u}_j, u_j \rangle.
\end{equation}
Here, we use \eqref{B-4}  and \eqref{B-6} to obtain
\begin{align*}
\begin{aligned} 
\langle{u_j, m\ddot{u}_j}\rangle &= -\gamma\langle{u_j, \dot{u}_j}\rangle \\
&+\frac{\kappa_0}{N} \sum_{k=1}^N \bigg(  \left\langle{ u_j, e^{\frac{\Omega_k -\Omega_j}{\gamma} t} u_k }\right\rangle -\left\langle{ e^{\frac{\Omega_k -\Omega_j}{\gamma} t} u_k, u_j}\right\rangle  \bigg) \| u_j \|^2  \\
&+\frac{\kappa_1}{N} \sum_{k=1}^N \bigg( \left\langle {u_j, e^{\frac{\Omega_k -\Omega_j}{\gamma} t} u_k} \right\rangle -\left\langle { e^{\frac{\Omega_k -\Omega_j}{\gamma} t} u_k, u_j} \right\rangle \bigg) \|u_j \|^2 -m \| \dot{u}_j \|^2.
\end{aligned}
\end{align*}
This yields
\[
\langle{u_j, m\ddot{u}_j}\rangle +\langle{m\ddot{u}_j, u_j}\rangle = \langle{u_j, m\ddot{u}_j}\rangle +\overline{\langle{u_j, m\ddot{u}_j}\rangle} = -\gamma \Big ( \langle{ u_j, \dot{u}_j }\rangle +\langle{ \dot{u}_j, u_j }\rangle  \Big) -2m \| \dot{u}_j \|^2.
\]
Now, we derive Gronwall's inequality for $\langle{u_j, \dot{u}_j}\rangle +\langle{\dot{u}_j, u_j}\rangle$:
\begin{align*}
m\frac{d}{dt}(\langle{u_j, \dot{u}_j}\rangle +\langle{\dot{u}_j, u_j}\rangle) &= 2m\|\dot{u}_j\|^2 +\langle{u_j, m\ddot{u}_j}\rangle +\langle{m\ddot{u}_j, u_j}\rangle =  -\gamma (\langle{u_j, \dot{u}_j}\rangle +\langle{\dot{u}_j, u_j}\rangle).
\end{align*}
Gronwall's lemma and initial conditions imply
\[
\frac{d}{dt} \|u_j \|^2 = \langle{u_j(t), \dot{u}_j(t) }\rangle +\langle{\dot{u}_j(t), u_j(t)}\rangle = e^{-\frac{\gamma}{m}t} ( \langle{ u_j^{in}, \dot{u}_j^{in} }\rangle +\langle{ \dot{u}_j^{in}, u_j^{in} }\rangle ) = 0, \quad \forall~t > 0.
\]
\end{proof}
\begin{lemma} \label{L2.2}
Suppose $\Omega_j$ satisfies
\begin{equation} \label{B-6-1}
\Omega^{\dagger} = -\Omega, \quad \Omega_j \equiv \Omega \quad \mbox{for all $j=1,\cdots,N$},
\end{equation}
where $\dagger$ denotes the Hermitian conjugate, and let $\{ z_j \}$ be a solution to \eqref{A-1}. Then, $u_j$ defined in \eqref{B-2} satisfies
\begin{equation*}
\begin{cases}
\displaystyle m \ddot{u}_j +\gamma \dot{u}_j = \kappa_0 \big( \| u_j \|^2 u_c -\left\langle{ u_c, u_j }\right\rangle u_j \big) +\kappa_1 \big( \left\langle{ u_j, u_c }\right\rangle -\left\langle{ u_c, u_j }\right\rangle \big) u_j -\frac{m \| \dot{u}_j \|^2}{\| u_j \|^2} u_j, \\
\displaystyle (u_j(0),  \dot{u}_j(0) ) = (u_j^{in}, \dot{u}_j^{in}), \quad \langle{u_j^{in}, \dot{u}_j^{in}}\rangle+\langle{\dot{u}_j^{in}, u_j^{in}}\rangle = 0,
\end{cases}
\end{equation*}
where $u_c := \frac{1}{N} \sum_{k=1}^N u_k$.
\end{lemma}
\begin{proof}
We substitute the relation \eqref{B-6-1} into \eqref{B-4} to get the desired estimate. 
\end{proof}
\begin{remark}
For a homogeneous ensemble with the common natural frequency $\Omega$, one can also see that $z_j$ satisfies 
\begin{align*}
\begin{cases}
\displaystyle m \Big( \dot{v}_j -\frac{1}{\gamma} \Omega v_j \Big) +\gamma v_j = \kappa_0(\langle{z_j, z_j}\rangle z_c-\langle{z_c, z_j}\rangle z_j)+\kappa_1(\langle{z_j, z_c}\rangle-\langle{z_c, z_j}\rangle)z_j -\frac{m \| v_j \|^2}{\| z_j \|^2}z_j, \\
\displaystyle v_j = \dot{z}_j -\frac{1}{\gamma} \Omega z_j, \quad z_j(0)=z_j^{in}, \quad v_j(0)=v_j^{in}, \quad \langle{z_j^{in}, v_j^{in}}\rangle+\langle{v_j^{in}, z_j^{in}}\rangle = 0.
\end{cases}
\end{align*}
\end{remark}
\subsection{Relation with other aggregation models} \label{sec:2.2}
In this subsection, we briefly discuss relations with other aggregation models with \eqref{A-1}. For a zero inertia and unit friction constant case:
\[ m = 0 \quad \mbox{and} \quad \gamma = 1 \]
system $\eqref{A-1}_1$ becomes the first-order LHS model \cite{H-P1, H-P2}:
\begin{equation} \label{B-8}
\displaystyle {\dot z}_j = \Omega_j z_j + \kappa_0 (\langle{z_j, z_j}\rangle z_c-\langle{z_c, z_j}\rangle z_j)+\kappa_1(\langle{z_j, z_c}\rangle-\langle{z_c, z_j}\rangle)z_j.
\end{equation}
Moreover, for the special case with $z_j =x_j \in \bbs^d \subset {\mathbb R}^{d+1}$ and $\kappa_1 = 0$, system \eqref{B-8} also reduces to the Lohe sphere model:
\begin{equation} \label{B-8-0}
\displaystyle {\dot x}_j = {\hat \Omega}_j x_j + \kappa_0 (\langle{x_j, x_j}\rangle x_c-\langle{x_c, x_j}\rangle x_j), \quad  {\hat \Omega}^\top = -  {\hat \Omega} \in \bbr^{(d+1)\times(d+1)}.
\end{equation}
The emergent dynamics of \eqref{B-8-0} has been extensively studied in \cite{C-C-H, C-H1, C-H2, C-H3, H-K-R2, H-K-L-N, H-K-L-N2, H-K-P-Z, Lo-1, Lo, Ma1, Ma2, Ol, S-H}. In the sequel, we  mainly discuss emergent behavior of the complex swarm sphere model \eqref{B-8}.  For this, we introduce new dependent variables: for state configuration $\{z_j \}$, 
\begin{equation} \label{B-8-1}
h_{ij} := \langle z_i, z_j \rangle, \quad  R_{ij} :=\mathrm{Re}(h_{ij}),\quad I_{ij} :=\mathrm{Im}(h_{ij})\quad\forall~i, j = 1, 2, \cdots, N.
\end{equation}
Note that
\[ h_{ij} = \langle z_i, z_j \rangle = 1 \quad \Longleftrightarrow \quad R_{ij}=1 \quad \mbox{and} \quad I_{ij}=0, \quad \forall ~i, j = 1, \cdots, N. \]
For a homogeneous ensemble with $\Omega_j = \Omega$, we expect the formation of complete aggregation which means 
\[ \lim_{t \to \infty} h_{ij} = 1. \]
Hence, it is natural to introduce a Lyapunov functional depending on the quantities:
\[ |1-h_{ij}| = \sqrt{(1-R_{ij})^2 + I_{ij}^2}. \] 
and we set 
\begin{align*}
\begin{aligned}
{\mathcal J}_{ij} &:=\sqrt[4]{(1-R_{ij})^2+I_{ij}^2}, \quad \forall i, j\in\{1, 2, \cdots, N\}, \\
 \mathcal J_M &:= \max_{i, j} \mathcal J_{ij} \quad \mbox{and} \quad {\mathcal D}(\Omega) := \max_{i,j} \| \Omega_i - \Omega_j \|_F.
\end{aligned}
\end{align*}
Now, we briefly summarize emergent behaviors of \eqref{B-8} without proofs.
\begin{theorem} \label{T2.1}
\emph{\cite{H-P0}}
The following assertions hold.
\begin{enumerate}
\item
(A homogeneous ensemble):~Suppose system parameters and initial data satisfy
\[ \kappa_0 > 2\kappa_1\geq0, \quad  \mathcal{D}(\Omega)=0, \quad  {\mathcal J}_{M}(0)  < \sqrt{1-\frac{2\kappa_1}{\kappa_0}}, \]
and let $\{z_j\}$ be the solution of \eqref{A-1} with initial data $\{z_j^{in}\}$. Then, there exists a positive constant ${\tilde \Lambda}$ such that 
\[
\mathcal J_M(t)\leq \mathcal J_M(0)\exp\left(-{\tilde \Lambda} t\right), \quad t > 0.
\]
\item
(A heterogeneous ensemble):~Suppose system parameters and initial data satisfy
\[  \kappa_1 \geq 0, \quad  \mathcal{D}(\Omega) > 0, \]
and let $\{z_j\}$ be a solution of \eqref{A-1} with the initial data $\{z_j^{in}\}.$ Then, one has a practical aggregation:
\[
\lim_{\kappa_0\to\infty}\limsup_{t\to\infty} \mathcal J_M(t)=0.
\]
\end{enumerate}
\end{theorem}

\vspace{0.5cm}

Before we close this section, we recall that how \eqref{B-8} can be further reduced to  the the Kuramoto model which is one of prototype examples for synchronization. We assume that the second coupling is absent and dimension in unit-dimensional: 
\[ \kappa_1 = 0, \quad d = 1. \]
In this case, we take the ansatz:
\begin{equation} \label{B-8-1}
z_j :=\begin{bmatrix}
\cos\theta_j\\
\sin\theta_j
\end{bmatrix}, \quad 
\Omega_j :=\begin{bmatrix}
0&-\nu_j \\
\nu_j &0
\end{bmatrix}, \quad j = 1, \cdots, N.
\end{equation}
We substitute \eqref{B-8-1} into \eqref{B-8} to obtain
\begin{align*}
\begin{aligned}
\dot{\theta}_i\begin{bmatrix}
-\sin\theta_j \\
\cos\theta_j
\end{bmatrix} &=\begin{bmatrix}
0&-\nu_i\\
\nu_i&0
\end{bmatrix}\begin{bmatrix}
\cos\theta_j \\
\sin\theta_j
\end{bmatrix} 
+\frac{\kappa}{N}\sum_{k=1}^N
\left( \begin{bmatrix}
\cos\theta_k\\
\sin\theta_k
\end{bmatrix}
-
\left\langle \begin{bmatrix}
\cos\theta_j \\
\sin\theta_j
\end{bmatrix}, 
\begin{bmatrix}
\cos\theta_k\\
\sin\theta_k
\end{bmatrix} \right\rangle\begin{bmatrix}
\cos\theta_j \\
\sin\theta_j
\end{bmatrix}\right) \\
&=\nu_i\begin{bmatrix}
-\sin\theta_j\\
\cos\theta_j
\end{bmatrix}+\frac{\kappa}{N}\sum_{k=1}^N
\begin{bmatrix}
\cos\theta_k-\sin\theta_k-(\cos\theta_i\cos\theta_k+\sin\theta_i\sin\theta_k)\cos\theta_j \\
\cos\theta_k+ \sin\theta_k-(\cos\theta_i\cos\theta_k+\sin\theta_i\sin\theta_k)\sin\theta_j
\end{bmatrix}\\
&=\nu_j \begin{bmatrix}
-\sin\theta_j \\
\cos\theta_j
\end{bmatrix}+\frac{\kappa}{N}\sum_{k=1}^N
\begin{bmatrix}
\cos\theta_k-\cos(\theta_j -\theta_k)\cos\theta_j \\
\sin\theta_k-\cos(\theta_j-\theta_k)\sin\theta_j
\end{bmatrix}\\
&=\nu_j \begin{bmatrix}
-\sin\theta_j \\
\cos\theta_j
\end{bmatrix}+\frac{\kappa}{N}\sum_{k=1}^N
\begin{bmatrix}
-\sin\theta_j \\
\cos\theta_j 
\end{bmatrix} \sin(\theta_k-\theta_j).
\end{aligned}
\end{align*}
We take an inner product the above relation with $(-\sin \theta_j, \cos \theta_j)^\top$, we obtain the Kuramoto model:
\[
\dot{\theta}_j =\nu_j +\frac{\kappa}{N}\sum_{k=1}^N\sin(\theta_k-\theta_j), \quad j = 1, \cdots, N.
\]
In summary, the LHS model generalizes the Lohe sphere model and Kuramoto model that were extensively studied in literature.
%%%%%%%%%%%%%%%%%%%%%%%%%%%%%%%%%%%
%
%  Section 3
%
%%%%%%%%%%%%%%%%%%%%%%%%%%%%%%%%%%%
\section{Characterization and instability of two distinguished states} \label{sec:3}
\setcounter{equation}{0}
In this section, we discuss the characterization and instability of two distinguished states for system \eqref{A-1} with zero frequency matrix and unit Hermitian sphere ${\mathbb H}{\mathbb S}^d $:
\[ \Omega_j \equiv 0, \quad \|z_j \| = 1, \quad j = 1, \cdots, N. \]
In this case, system \eqref{A-1} takes a much simpler form:
\begin{equation} \label{C-1}
m \ddot{z}_j +\gamma \dot{z}_j = \kappa_0 \big(z_c -\left\langle{ z_c, z_j }\right\rangle z_j \big) +\kappa_1 \big( \left\langle{ z_j, z_c }\right\rangle -\left\langle{ z_c, z_j }\right\rangle \big) z_j -m \| \dot{z}_j \|^2 z_j.
\end{equation}
This can also be rewritten as a first-order system by introducing an auxiliary variable $w_j = {\dot z}_j$:
\begin{align}\label{C-2}
\begin{aligned}
 {\dot z}_j &= w_j, \\
 \dot{w}_j &= -\frac{\gamma}{m} w_j + \frac{\kappa_0}{m} \big( z_c -\left\langle{ z_c, z_j }\right\rangle z_j \big) + \frac{\kappa_1}{m} \big( \left\langle{ z_j, z_c }\right\rangle -\left\langle{ z_c, z_j }\right\rangle \big) z_j -\| w_j \|^2  z_j.
\end{aligned}
\end{align}
\subsection{Characterization of equilibria} \label{sec:3.1} 
Note that the algebraic equilibrium system associated with \eqref{C-2}:
\begin{equation}\label{C-3}
\begin{cases}
\displaystyle w_j = 0, \\
\displaystyle -\gamma w_j + \kappa_0 \big( z_c -\left\langle{ z_c, z_j }\right\rangle z_j \big) + \kappa_1 \big( \left\langle{ z_j, z_c }\right\rangle -\left\langle{ z_c, z_j }\right\rangle \big) z_j - m \| w_j \|^2  z_j = 0.
\end{cases}
\end{equation}
\begin{proposition}\label{P3.1}
Let $\{(z^e_j, w_j^e) \}$ be an equilibrium solution of \eqref{C-2} if and only if $(z_j^e, w_j^e)$ is a constant state satisfying 
\begin{equation*}
w^e_j= 0, \quad z_c^e = \langle{z_c^e,z_j^e}\rangle z_j^e, \quad \forall ~j = 1, \cdots, N,
\end{equation*}
where $z_c^e = \frac{1}{N}\sum_{j=1}^N z_j^e$. 
\end{proposition}
\begin{proof} 
$(\Longrightarrow$ part):~Suppose $\{(z^e_j, w_j^e) \}$ is an equilibrium state. Then, it satisfies
\begin{equation}\label{C-5}
w_j^e = 0, \quad 0 = \kappa_0(z_c^e-\langle{z_c^e, z_j^e}\rangle z_j^e)+\kappa_1(\langle{z_j^e, z_c^e}\rangle-\langle{z_c^e, z_j^e}\rangle) z_j^e.
\end{equation}
Now, we use the relation $\|z_j^e \| = 1$ to see
\begin{equation*}
0 = \langle z_j^e, \eqref{C-5}_2 \rangle = (\kappa_0 + \kappa_1) \Big( \langle{z_j^e, z_c^e}\rangle-\langle{z_c^e, z_j^e}\rangle \Big).
\end{equation*}
Since $\kappa_0 > 0$ and $\kappa_1 \geq 0$, one has 
\begin{align}\label{C-7}
0 = \langle{z_j^e, z_c^e}\rangle-\langle{z_c^e, z_j^e}\rangle.
\end{align}
Then, we substitute \eqref{C-7} into $\eqref{C-5}_2$ and  use $\kappa_0 > 0$ to get
\[    z_c^e = \langle{z_c^e,z_j^e}\rangle z_j^e. \]

\vspace{0.5cm}

\noindent $(\Longleftarrow$ part):~Suppose that a constant state $(z^e_j, w^e_j)$ satisfies relations:
\begin{equation} \label{C-8}
z^e_c = \langle{z^e_c,z^e_j}\rangle z^e_j \quad \forall~t \geq 0 \quad \mbox{and} \quad w^e_j = 0, \quad j = 1, \cdots, N.
\end{equation}
We use the relation $\eqref{C-8}_1$ and $\|z_j^e \| = 1$ to find
\begin{align*}
\begin{aligned}
& \left\langle{ z^e_j, z^e_c }\right\rangle -\left\langle{ z^e_c, z^e_j }\right\rangle \\
& \hspace{0.5cm} = \Big \langle z^e_j,  \langle{z^e_c,z^e_j}\rangle z^e_j \Big \rangle - \Big \langle  \langle{z^e_c,z^e_j}\rangle z^e_j, z_j^e \Big \rangle =  \langle z^e_c, z^e_j \rangle \langle z^e_j, z^e_j \rangle - \overline{\langle z^e_c, z^e_j \rangle}  \langle z^e_j, z^e_j \rangle \\
&  \hspace{0.5cm} = \langle z^e_c, z^e_j \rangle \langle z^e_j, z^e_j \rangle - \langle z^e_j, z^e_c \rangle  \langle z^e_j, z^e_j \rangle = -\Big( \left\langle{ z^e_j, z^e_c }\right\rangle -\left\langle{ z^e_c, z^e_j }\right\rangle \Big).
\end{aligned}
\end{align*}
This yields
\begin{equation} \label{C-10}
\left\langle{ z^e_j, z^e_c }\right\rangle -\left\langle{ z^e_c, z^e_j }\right\rangle  = 0, \quad \forall~t \geq 0.
\end{equation}
Finally, the relations $\eqref{C-8}$ and \eqref{C-10} satisfy the equilibrium system \eqref{C-3}.
%\begin{align*}
%m \dot{w}_j +\gamma w_j = -m \|w_j \|^2 z_j, \quad t > 0.
%\end{align*}
%This implies
%\begin{align*}
%m\frac{d}{dt} \|w_j(t) \|^2 & = \langle{w, m {\dot w}_j}\rangle +\langle{m {\dot w}_j, w_j}\rangle \\
%& = -2\gamma \| w_j \|^2 -m \| w_j \|^2 ( \langle{w_j, z_j}\rangle +\langle{z_j, w_j}\rangle ) =  -2\gamma \| w_j \|^2.
%\end{align*}
%By Gronwall's lemma, one has 
%\begin{align*}
%\| w_j(t) \|^2 = e^{-\frac{2\gamma}{m} t} \|w_j^e(0) \|^2 = 0, \quad w_j  = 0.
%\end{align*}
%Therefore, we have
%\[ {\dot z}_j  = 0, \quad {\dot w}_j = 0, \quad \forall~t > 0, \quad  j  = 1, \cdots, N. \]
\end{proof}
Next, we introduce an order parameter $\rho$ which measures the degree of aggregation. For a given configuration $\{(z_j, w_j) \}$, we set 
\begin{equation} \label{C-10-1}
\rho := \Big \| \frac{1}{N} \sum_{j} z_j \Big \|, \qquad \rho^\infty := \lim_{t \to \infty} \rho(t) \quad \mbox{if it does exist}.
\end{equation}
Then, $\rho = 0, 1$ denote the incoherent state and completely aggregated state, respectively. \newline

As a corollary of Proposition \ref{P3.1}, we show that equilibrium with positive $\rho$ is either completely aggregated state or a bi-polar state.
\begin{corollary} \label{C3.1}
For $d = 0$, let $(z_j^e, w_j^e)$ be an equilibrium solution with $\rho > 0$ and $|z_j^e| =1$. Then, one has 
\[    \frac{z_j^e}{z_c^e} \in \bbr. \]
\end{corollary}
\begin{proof}
 Let $\{(z_j^e, w_j^e) \}$ be an equilibrium state with $\rho > 0$ and $|z_j^e| =1$. Then, by Proposition \ref{P3.1}, one has 
\begin{equation} \label{C-11}
z_c^e=\langle z_c^e, z_j^e\rangle z_j^e, \quad j = 1, \cdots, N. 
\end{equation}
On the other hand, since $\rho = |z_c^e| >0$ we write down $z_c$ and $z_j$ as polar forms:
\begin{equation} \label{C-12}
z_j^e=e^{\mathrm{i}\theta_j} \quad \mbox{and} \quad z_c^e=\rho e^{\mathrm{i}\phi}, \quad \rho > 0.
\end{equation}
Now, we substitute \eqref{C-12} into \eqref{C-11} to get 
\[
\rho e^{\mathrm{i}\phi}=\rho e^{\mathrm{i}(2\theta_j-\phi)}.
\]
This yields
\[
e^{2\mathrm{i}\phi}=e^{2\mathrm{i}\theta_j}, \quad j = 1, \cdots, N.
\]
Hence, one has
\[
\mbox{either}~\theta_j = \phi \quad\text{or}\quad \theta_j =  \phi+\pi, \quad j = 1, \cdots, N.
\]
Thus,
\[ \frac{z_j^e}{z_c^e} = \frac{1}{\rho} e^{{\mathrm i} (\theta_j - \phi)} \in \Big \{ \frac{1}{\rho},~-\frac{1}{\rho}  \Big \}.     \]
\end{proof}

\subsection{Instability of two distinguished states} \label{sec:3.2}
In this subsection, we study linear instabilities of two distinguished state ``{\it bi-polar state and incoherence state $(\rho = 0)$}". 
\begin{align}\label{D-32}
\begin{cases}
\displaystyle {\dot z}_j = w_j, \\
\displaystyle \dot{w}_j = -\frac{\gamma}{m} w_j + \frac{\kappa_0}{m} \big(z_c -\left\langle{ z_c, z_j }\right\rangle z_j \big) + \frac{\kappa_1}{m} \big( \left\langle{ z_j, z_c }\right\rangle -\left\langle{ z_c, z_j }\right\rangle \big) z_j -\| w_j \|^2 z_j.
\end{cases}
\end{align}
In the sequel, we consider $z_j$ and $w_j$ as real vectors in $\bbr^{2d+2}$. In other words, let $x^\alpha_j, y^\alpha_j, a^\alpha_j, b^\alpha_j \in \bbr$ be given as follows:
\begin{align} \label{C-1-1}
z^\alpha_j = x^\alpha_j +\mathrm iy^\alpha_j, \quad w^\alpha_j = a^\alpha_j +\mathrm ib^\alpha_j, \quad j = 1, \cdots, N, \quad \alpha = 1, \cdots, d+1,
\end{align}
where $z_j^\alpha$ and $w_j^\alpha$ are $a$-th component of $z_j$ and $w_j$, respectively. We rewrite \eqref{D-32} using \eqref{C-1-1}:
\begin{align*}
\begin{cases}
\dot x_j = a_j, \quad \dot y_j = b_j, \vspace{.2cm} \\
\displaystyle \dot a_j = -\frac{\gamma}{m}a_j +\frac{\kappa_0}{m} \big[ x_c -\big( \langle x_c, x_j \rangle +\langle y_c, y_j \rangle \big) x_j +\big( \langle x_c, y_j \rangle -\langle y_c, x_j \rangle \big) y_j \big]\vspace{.2cm} \\
\displaystyle \hspace{.7cm} -\frac{2\kappa_1}{m} \big( \langle y_c, x_j \rangle -\langle x_c, y_j \rangle \big) y_j -\big( \|a_j\|^2 +\|b_j\|^2 \big) x_j, \vspace{.2cm} \\
\displaystyle \dot b_j = -\frac{\gamma}{m} b_j +\frac{\kappa_0}{m} \big[ y_c -\big( \langle x_c, x_j \rangle +\langle y_c, y_j \rangle \big) y_j -\big( \langle x_c, y_j \rangle -\langle y_c, x_j \rangle \big) x_j \big] \vspace{.2cm} \\
\displaystyle \hspace{.7cm} +\frac{2\kappa_1}{m} \big( \langle y_c, x_j \rangle -\langle x_c, y_j \rangle \big) x_j -\big( \|a_j\|^2 +\|b_j\|^2 \big) y_j,
\end{cases}
\end{align*}
For stability analysis, we also define
\begin{align*}
\mathcal I := (x_1, \cdots, x_N, y_1, \cdots, y_N, a_1, \cdots, a_N, b_1, \cdots, b_N) = (c_1, \cdots, c_{4N}) \in \bbr^{4(d+1)N},
\end{align*}
and consider the following Jacobian matrix at equilibrium $\mathcal I^e$:
\begin{align*}
\mathcal{M}:=\frac{\partial\dot{\mathcal{I}}}{\partial\mathcal{I}} \bigg|_{\mathcal I = \mathcal I^e} = (\mathcal M_{ij})_{1\leq i,j \leq 4}, \quad \mathcal M_{ij} := \frac{\partial (\dot c_{(i-1)N+1}, \cdots, \dot c_{(i-1)N+N} )}{\partial (c_{(j-1)N+1}, \cdots, c_{(j-1)N+N})} \bigg|_{\mathcal I = \mathcal I^e}.
\end{align*}
By direct calculations, one has 
\begin{align*}
& \mathcal M_{11} = \mathcal M_{12} = \mathcal M_{14} = \mathcal M_{21} = \mathcal M_{21} = \mathcal M_{23} = \mathcal M_{34} = \mathcal M_{43} = O_{(d+1)N}, \\
& \mathcal M_{13} = \mathcal M_{24} = I_{(d+1)N}, \quad \mathcal M_{33} = \mathcal M_{44} = -\frac{\gamma}{m} I_{(d+1)N},
\end{align*}
where we used $w_j = 0$ at equilibrium to calculate $\mathcal M_{33}$ and $\mathcal M_{44}$. Hence, $\mathcal M$ has following form:
\begin{align*}
\mathcal{M}:=\frac{\partial\dot{\mathcal{I}}}{\partial\mathcal{I}} = 
\left(\begin{matrix}
O_{2(d+1)N} & I_{2(d+1)N} \vspace{.2cm} \\
\mathcal M_s & -\frac{\gamma}{m}I_{2(d+1)N}
\end{matrix}\right), \quad \mathcal M_s = \left(\begin{matrix}
\mathcal M_{31} & \mathcal M_{32} \vspace{.1cm} \\
\mathcal M_{41} & \mathcal M_{42}
\end{matrix}\right),
\end{align*}
We use the fact that
\begin{align*}
\left|\begin{matrix}
A & B \\
C & D
\end{matrix}\right| = \det(A-BD^{-1}C)\det(D)
\end{align*}
to observe the relation between eigenvalues of $\mathcal{M}$ and $\mathcal M_s$:
\begin{align*}
\det\left(\mathcal{M}-\lambda I_{4(d+1)N}\right) =
\left|\begin{matrix}
-\lambda I_{2(d+1)N} & I_{2(d+1)N}\\
\mathcal M_s & -\left(\frac{\gamma}{m}+\lambda\right)I_{2(d+1)N}
\end{matrix}\right| = \det\left(\lambda\left(\frac{\gamma}{m}+\lambda\right)I_{2(d+1)N} -\mathcal M_s \right).
\end{align*}
It follows from the above equation that if $\lambda_0$ is the eigenvalue of $\mathcal M_s$, then $\lambda$ satisfying
\begin{align*}
\lambda_0 = \lambda\bigg( \frac{\gamma}{m} +\lambda \bigg)
\end{align*}
is also an eigenvalue of $\mathcal{M}$. \newline

Suppose that $\mathcal M_s$ has an eigenvalue $\lambda_p$, whose real part is positive. Then, one can see
\begin{align*}
\mbox{Re} \lambda_p = \mbox{Re}\lambda\bigg( \frac{\gamma}{m} +\mbox{Re}\lambda \bigg) -( \mbox{Im}\lambda )^2 \iff \mbox{Re}\lambda = \frac{-\gamma \pm \sqrt{\gamma^2+4m\big( \mbox{Re} \lambda_p +( \mbox{Im}\lambda )^2 \big)}}{2m},
\end{align*}
which implies that $\mathcal M$ has an eigenvalue which has positive real part. More precisely,  one has
\[
\frac{-\gamma +\sqrt{\gamma^2+4m\big( \mbox{Re} \lambda_p +( \mbox{Im}\lambda )^2 \big)}}{2m} > 0.
\]
Hence, we need further estimates on $\mathcal M_s$. We calculate components of $\mathcal M_s$ one by one. 
\begin{align*}
& \mathcal M_{31} = (A^1_{jk})_{j, k}, \quad \mathcal M_{32} = (A^2_{jk})_{j, k}, \quad \mathcal M_{41} = (A^3_{jk})_{j, k}, \quad \mathcal M_{42} = (A^4_{jk})_{j, k}, \quad 1\leq j,k \leq N, \\
& A^1_{jk} := \frac{\partial \dot a_j}{\partial x_k} = \bigg( \frac{\partial \dot a_j^\alpha}{\partial x_k^\beta} \bigg)_{\alpha, \beta}, \quad A^2_{jk} := \frac{\partial \dot a_j}{\partial y_k} = \bigg( \frac{\partial \dot a_j^\alpha}{\partial y_k^\beta} \bigg)_{\alpha, \beta}, \\
&  A^3_{jk} := \frac{\partial \dot b_j}{\partial x_k} = \bigg( \frac{\partial \dot b_j^\alpha}{\partial x_k^\beta} \bigg)_{\alpha, \beta}, \quad A^4_{jk} := \frac{\partial \dot b_j}{\partial y_k} = \bigg( \frac{\partial \dot b_j^\alpha}{\partial y_k^\beta} \bigg)_{\alpha, \beta}, \quad 1 \leq \alpha, \beta \leq d+1.
\end{align*}

\vspace{.2cm}

\noindent More precisely, we have
\begin{align*}
\frac{\partial \dot a_j^\alpha}{\partial x_k^\beta} & = \frac{\kappa_0}{m}\frac{\partial}{\partial x_k^\beta} \Big( x_c^\alpha -\big( \langle x_c, x_j \rangle +\langle y_c, y_j \rangle \big) x_j^\alpha +\big( \langle x_c, y_j \rangle -\langle y_c, x_j \rangle \big) y_j^\alpha \Big) \\
& \hspace{.3cm} -\frac{2\kappa_1}{m} \frac{\partial}{\partial x_k^\beta} \Big( \big( \langle y_c, x_j \rangle -\langle x_c, y_j \rangle \big) y_j^\alpha \Big) -\frac{\partial}{\partial x_k^\beta} \Big( \big( \|a_j\|^2 +\|b_j\|^2 \big) x_j^\alpha \Big) \\
& = \frac{\kappa_0}{m} \bigg( \frac{\delta_{\alpha\beta}}{N} -\bigg( \bigg\langle \frac{e^\beta}{N}, x_j \bigg\rangle +\langle x_c, \delta_{jk}e^\beta \rangle \bigg) x_j^\alpha -(\langle x_c, x_j \rangle +\langle y_c, y_j \rangle) \frac{\partial x_j^\alpha}{\partial x_k^\beta} \\
& \hspace{0.3cm} +\bigg( \bigg\langle \frac{e^\beta}{N}, y_j \bigg\rangle -\langle y_c, \delta_{jk}e^\beta \rangle \bigg) y_j^\alpha \bigg) -\frac{2\kappa_1}{m} \bigg( \bigg( \langle y_c, \delta_{jk}e^\beta \rangle -\bigg\langle \frac{e^\beta}{N}, y_j \bigg\rangle \bigg) y_j^\alpha \bigg) \\
& = \frac{\kappa_0}{m} \bigg( \frac{\delta_{\alpha\beta}}{N} -\bigg( \frac{x_j^\beta}{N} +\delta_{jk}x_c^\beta \bigg) x_j^\alpha -(\langle x_c, x_j \rangle +\langle y_c, y_j \rangle) \frac{\partial x_j^\alpha}{\partial x_k^\beta} +\bigg( \frac{y_j^\beta}{N} -\delta_{jk}y_c^\beta \bigg) y_j^\alpha \bigg) \\
& \hspace{.3cm} -\frac{2\kappa_1}{m} \bigg( \bigg( \delta_{jk}y_c^\beta -\frac{y_j^\beta}{N} \bigg) y_j^\alpha \bigg).
\end{align*}
Similarly, one can see
\begin{align*}
\frac{\partial \dot a_j^\alpha}{\partial y_k^\beta} &= \frac{\kappa_0}{m}\frac{\partial}{\partial y_k^\beta} \Big( x_c^\alpha -\big( \langle x_c, x_j \rangle +\langle y_c, y_j \rangle \big) x_j^\alpha +\big( \langle x_c, y_j \rangle -\langle y_c, x_j \rangle \big) y_j^\alpha \Big) \\
& \hspace{.3cm} -\frac{2\kappa_1}{m} \frac{\partial}{\partial y_k^\beta} \Big( \big( \langle y_c, x_j \rangle -\langle x_c, y_j \rangle \big) y_j^\alpha \Big) -\frac{\partial}{\partial y_k^\beta} \Big( \big( \|a_j\|^2 +\|b_j\|^2 \big) x_j^\alpha \Big) \\
& = -\frac{\kappa_0}{m} \bigg( \bigg( \bigg\langle \frac{e^\beta}{N}, y_j \bigg\rangle +\langle y_c, \delta_{jk}e^\beta \rangle \bigg) x_j^\alpha \\
& \hspace{0.3cm}  -\bigg( \langle x_c, \delta_{jk}e^\beta \rangle -\bigg\langle \frac{e^\beta}{N}, x_j \bigg\rangle \bigg) y_j^\alpha -\big( \langle x_c, y_j \rangle -\langle y_c, x_j \rangle \big) \frac{\partial y_j^\alpha}{\partial y_k^\beta} \bigg) \\
& \hspace{.3cm} -\frac{2\kappa_1}{m} \bigg( \bigg( \bigg\langle \frac{e^\beta}{N}, x_j \bigg\rangle -\langle x_c, \delta_{jk}e^\beta \rangle \bigg) y_j^\alpha +\big( \langle y_c, x_j \rangle -\langle x_c, y_j \rangle \big) \frac{\partial y_j^\alpha}{\partial y_k^\beta} \bigg) \\
& = -\frac{\kappa_0}{m} \bigg( \bigg( \frac{y_j^\beta}{N} +\delta_{jk}y_c^\beta \bigg) x_j^\alpha -\bigg( \delta_{jk}x_c^\beta -\frac{x_j^\beta}{N} \bigg) y_j^\alpha -\big( \langle x_c, y_j \rangle -\langle y_c, x_j \rangle \big) \frac{\partial y_j^\alpha}{\partial y_k^\beta} \bigg) \\
& \hspace{.3cm} -\frac{2\kappa_1}{m} \bigg( \bigg( \frac{x_j^\beta}{N} -\delta_{jk}x_c^\beta \bigg) y_j^\alpha +\big( \langle y_c, x_j \rangle -\langle x_c, y_j \rangle \big) \frac{\partial y_j^\alpha}{\partial y_k^\beta} \bigg).
\end{align*}
In the same way, we can observe
\begin{align*}
\frac{\partial \dot b_j^\alpha}{\partial x_k^\beta} &= -\frac{\kappa_0}{m} \bigg( \bigg( \frac{x_j^\beta}{N} +\delta_{jk}x_c^\beta \bigg) y_j^\alpha -\bigg( \delta_{jk}y_c^\beta -\frac{y_j^\beta}{N} \bigg) x_j^\alpha -\big( \langle y_c, x_j \rangle -\langle x_c, y_j \rangle \big) \frac{\partial x_j^\alpha}{\partial x_k^\beta} \bigg) \\
& \hspace{.3cm} -\frac{2\kappa_1}{m} \bigg( \bigg( \frac{y_j^\beta}{N} -\delta_{jk}y_c^\beta \bigg) x_j^\alpha +\big( \langle x_c, y_j \rangle -\langle y_c, x_j \rangle \big) \frac{\partial x_j^\alpha}{\partial x_k^\beta} \bigg), \\
\frac{\partial \dot b_j^\alpha}{\partial y_k^\beta} &= \frac{\kappa_0}{m} \bigg( \frac{\delta_{\alpha\beta}}{N} -\bigg( \frac{y_j^\beta}{N} +\delta_{jk}y_c^\beta \bigg) y_j^\alpha -(\langle x_c, x_j \rangle +\langle y_c, y_j \rangle) \frac{\partial y_j^\alpha}{\partial y_k^\beta} +\bigg( \frac{x_j^\beta}{N} -\delta_{jk}x_c^\beta \bigg) x_j^\alpha \bigg) \\
& \hspace{.3cm} -\frac{2\kappa_1}{m} \bigg( \bigg( \delta_{jk}x_c^\beta -\frac{x_j^\beta}{N} \bigg) x_j^\alpha \bigg).
\end{align*}

\vspace{.2cm}

In what follows, we study stability of two distinguished states. \newline

\noindent$\bullet$ (Instability of an incoherence state): Since the trace of a matrix is equal to the sum of its eigenvalues, we observe
\begin{align*}
 \mbox{Tr} \mathcal M_s &= \mbox{Tr} \mathcal M_{31} +\mbox{Tr} \mathcal M_{42} = \sum_{j = 1}^N \mbox{Tr} A^1_{jj} +\sum_{j = 1}^N \mbox{Tr} A^4_{jj} = \sum_{j = 1}^N \sum_{\alpha = 1}^{d+1} \frac{\partial \dot a_j^\alpha}{\partial x_j^\alpha} +\sum_{j = 1}^N \sum_{\alpha = 1}^{d+1} \frac{\partial \dot b_j^\alpha}{\partial y_j^\alpha} \\
& = \frac{\kappa_0}{m} \sum_{j=1}^N\sum_{\alpha=1}^{d+1} \bigg( \frac{1}{N} -\bigg( \frac{x_j^\alpha}{N} +x_c^\alpha \bigg) x_j^\alpha +\langle x_c, x_j \rangle +\bigg( \frac{y_j^\alpha}{N} -y_c^\alpha \bigg) y_j^\alpha \bigg) \\
& \hspace{.2cm} -\frac{2\kappa_1}{m} \sum_{j=1}^N\sum_{\alpha=1}^{d+1} \bigg( \bigg( y_c^\alpha -\frac{y_j^\alpha}{N} \bigg) y_j^\alpha \bigg) \\
& \hspace{.2cm} +\frac{\kappa_0}{m} \sum_{j=1}^N\sum_{\alpha=1}^{d+1} \bigg( \frac{1}{N} -\bigg( \frac{y_j^\alpha}{N} +y_c^\alpha \bigg) y_j^\alpha +\langle y_c, y_j \rangle +\bigg( \frac{x_j^\alpha}{N} -x_c^\alpha \bigg) x_j^\alpha \bigg) \\
& \hspace{.2cm} -\frac{2\kappa_1}{m} \sum_{j=1}^N\sum_{\alpha=1}^{d+1} \bigg( \bigg( x_c^\alpha -\frac{x_j^\alpha}{N} \bigg) x_j^\alpha \bigg) \\
& = \frac{\kappa_0}{m} \sum_{j=1}^N\sum_{\alpha=1}^{d+1} \bigg( \frac{1}{N} -\frac{\big(x_j^\alpha\big)^2}{N} +\frac{\big( y_j^\alpha \big)^2}{N} \bigg) +\frac{2\kappa_1}{m} \sum_{j=1}^N\sum_{\alpha=1}^{d+1} \frac{\big(y_j^\alpha\big)^2}{N} \\
& \hspace{.2cm} +\frac{\kappa_0}{m} \sum_{j=1}^N\sum_{\alpha=1}^{d+1} \bigg( \frac{1}{N} -\frac{\big(y_j^\alpha\big)^2}{N} +\frac{\big(x_j^\alpha\big)^2}{N} \bigg) +\frac{2\kappa_1}{m} \sum_{j=1}^N\sum_{\alpha=1}^{d+1} \frac{\big( x_j^\alpha\big)^2}{N} \\
& = \frac{2(d+1)\kappa_0}{m} +\frac{2\kappa_1}{m} > 0,
\end{align*}
where we used
\begin{align*}
x_c = y_c = 0.
\end{align*}
%\begin{align*}
%\mbox{tr}C &= \sum_{j=1}^N\mbox{tr}E_{jj} = \sum_{j=1}^N\sum_{k=1}^d\frac{\partial \dot{v}_j^k}{\partial z_j^k} \\
%& = \frac{\kappa_0}{m}\sum_{k=1}^d\sum_{j=1}^N\left[\frac{1}{N}-\left(\frac{1}{N}z_j^k+\overline{z_c^k}\right)z_j^k-\langle{z_c,z_j}\rangle\right] \\
%& \hspace{.5cm} +\frac{\kappa_1}{m}\sum_{k=1}^d\sum_{j=1}^N\left(z_c^k+\frac{1}{N}\overline{z_j^k}-\frac{1}{N}z_j^k-\overline{z_c^k}\right)z_j^k +\frac{\kappa_1}{m}\sum_{k=1}^d\sum_{j=1}^N(\langle{z_j, z_c}\rangle-\langle{z_c, z_j}\rangle) \\
%&= \sum_{k=1}^d\sum_{j=1}^N\left[\frac{\kappa_0}{mN}\left\{1-\left(z_j^k\right)^2\right\}+\frac{\kappa_1}{mN}\left(\overline{z_j^k}-z_j^k\right)z_j^k\right] \\
%&= \frac{d\kappa_0+\kappa_1}{m}-\frac{\kappa_0+\kappa_1}{mN}\sum_{k=1}^d\sum_{j=1}^N\left(z_j^k\right)^2.
%\end{align*}
%This implies
%\begin{align*}
%\mbox{Re}\left(\mbox{tr}C\right) &= \frac{d\kappa_0+\kappa_1}{m}-\frac{\kappa_0+\kappa_1}{mN}\sum_{k=1}^d\sum_{j=1}^N\left[\left(\mbox{Re}z_j^k\right)^2-\left(\mbox{Im}z_j^k\right)^2\right] \\
%&\geq\frac{d\kappa_0+\kappa_1}{m}-\frac{\kappa_0+\kappa_1}{mN}\sum_{j=1}^N\sum_{k=1}^d\left[\left(\mbox{Re}z_j^k\right)^2+\left(\mbox{Im}z_j^k\right)^2\right] \\
%&= \frac{d\kappa_0+\kappa_1}{m}-\frac{\kappa_0+\kappa_1}{m} = \frac{\kappa_0(d-1)}{m} >0.
%\end{align*}
Hence, $\mathcal M_s$ has at least one eigenvalue whose real part is positive and so does $\mathcal{M}$. Therefore, we can conclude that the incoherence state is unstable.

\vspace{.2cm}

\noindent$\bullet$ (Instability of bi-polar state): Suppose that there exists a point $z$ and integer $n$ such that
\begin{align*}
\|z\| = 1,\quad 1\leq n < \left\lfloor \frac{N}{2} \right\rfloor,\quad z_i=-z, \quad z_j = z, \quad 1\leq i \leq n, \quad n+1 \leq j \leq N.
\end{align*}
Without loss of generality, we can assume that $z = z^\infty = (0,\cdots,1)$ using the rotational symmetry of the LHS with inertia. Then, we have 
\begin{align*}
z_c = \left(0,\cdots,0,\frac{N-2n}{N}\right).
\end{align*}
Then, further calculation yields
\[
\frac{\partial \dot a_j^\alpha}{\partial x_k^\beta} = \frac{\kappa_0}{m} \bigg( \frac{\delta_{\alpha\beta}}{N} -\bigg( \frac{x_j^\beta}{N} +\delta_{jk}x_c^\beta \bigg) x_j^\alpha -\langle x_c, x_j \rangle \frac{\partial x_j^\alpha}{\partial x_k^\beta}  \bigg) \quad \mbox{and} \quad \frac{\partial \dot b_j^\alpha}{\partial x_k^\beta} = 0.
\]
We observe $(n+1)(d+1)$-th column of $\mathcal M_s$:
\begin{align*}
\mathcal M_s \tilde e_{(n+1)(d+1)} &= \left(\left(A^1_{1,n+1}\right)_{1,d+1},\cdots,\left(A^1_{N,n+1}\right)_{d+1,d+1}, 0, \cdots, 0\right)^\top = \frac{2\kappa_0(N-2n)}{mN} \tilde e_{(n+1)(d+1)},
\end{align*}
where $\{ \tilde e_\alpha \}_{\alpha = 1}^{2(d+1)N}$ is a standard basis on $\bbr^{2(d+1)N}$.
%\vspace{0.5cm}
%
%\noindent$\diamond$ Case A ($j=k$):
%\begin{align*}
%\frac{\partial \dot{v}_j^r}{\partial z_j^s} &= \frac{\kappa_0}{m}\left[\frac{\delta_{rs}}{N}-\left(\frac{1}{N}z_j^s+\overline{z_c^s}\right)z_j^r-\delta_{rs}\langle{z_c,z_j}\rangle\right] \\
%&= \begin{cases}
%\displaystyle \frac{2\kappa_0(N-2n)}{mN} & \quad r=s=d,\ 1\leq j\leq n, \vspace{0.1cm} \vspace{.2cm}\\
%\displaystyle \frac{2\kappa_0(2n-N)}{mN} & \quad r=s=d,\ n+1\leq j\leq N, \vspace{0.1cm} \vspace{.2cm}\\
%\displaystyle \frac{\kappa_0(1-2n+N)}{mN} & \quad r=s\ne d,\ 1\leq j\leq n, \vspace{0.1cm} \vspace{.2cm}\\
%\displaystyle \frac{\kappa_0(1-N+2n)}{mN} & \quad r=s\ne d,\ n+1\leq j\leq N, \vspace{0.1cm} \vspace{.2cm}\\
%0 & \quad r\ne s.
%\end{cases}
%\end{align*}
%\noindent$\diamond$ Case B ($j\ne k$):
%\begin{align*}
%& \frac{\partial \dot{v}_j^r}{\partial z_k^s}=\frac{\kappa_0}{m}\left[\frac{\delta_{rs}}{N}-\frac{1}{N}z_j^sz_j^r\right] = \begin{cases}
%\displaystyle 0 & r=s=d\ \mbox{or}\ r\ne s, \vspace{0.2cm} \\
%\displaystyle \frac{\kappa_0}{mN} & r=s\ne d.
%\end{cases}
%\end{align*}
%We combine above to cases to obtain
%\begin{align*}
%Ce_{d} &= \left[\left(E_{11}\right)_{1d},\cdots,\left(E_{11}\right)_{dd},\cdots,\left(E_{N1}\right)_{1d},\cdots,\left(E_{N1}\right)_{dd}\right]^T = \frac{2\kappa_0(N-2n)}{mN}e_d, \\
%Ce_{(n+1)d} &= \left[\left(E_{1,n+1}\right)_{1d},\cdots,\left(E_{N,n+1}\right)_{dd}\right]^T = \frac{2\kappa_0(2n-N)}{mN}e_{(n+1)d}.
%\end{align*}
Since $n$ is smaller than $\lfloor N/2 \rfloor$, we have positive eigenvalue. Therefore, we can conclude that the bipolar state is unstable.

\section{Presentation of main results}\label{sec:4}
\setcounter{equation}{0}
In this section, we briefly summarize frameworks for the emergent dynamics of the second-order extension of the first-order LHS model.

\subsection{Complete aggregation} \label{sec:4.1}
In this subsection, we present an emergent dynamics of the homogeneous ensemble with the same natural frequency matrix $\Omega_j = \Omega$. For this, we set 
\begin{equation*}
\begin{cases} 
\displaystyle  h_{ij} := \langle{ z_i, z_j \rangle}, \quad g_{ij} := 1-h_{ij}, \quad 1 \leq i,j \leq N, \vspace{.2cm} \\
\displaystyle {\mathcal G} := \frac{1}{N^2}\sum_{i, j=1}^N | g_{ij} |^2, \quad {\mathcal R}_1(\dot Z) := \max_{j} \| \dot{z}_j \|^2, \quad {\mathcal R}_2(Z) := \max_{j} | \langle{ z_c, z_{j} \rangle} -\langle{ z_{j}, z_c \rangle} |^2, \vspace{.2cm} \\
\displaystyle M_1 := \max \bigg\{ \|w_1^{in} \|, \cdots, \|w_N^{in} \|, \frac{2(\kappa_0 +\kappa_1)}{\gamma} \bigg\}, \quad \nu_1 := \frac{\gamma +\sqrt{\gamma^2 -16m\kappa_0\delta}}{2m}.
\end{cases}
\end{equation*}
Then, it is easy to see that 
\[ |h_{ij}| \leq 1, \quad |g_{ij}| \leq 2, \quad \overline{h_{ij}} = h_{ji} \quad \mbox{and} \quad \overline{g_{ij}} = g_{ji}, \quad 1 \leq i, j \leq N. \]

Now, we set up two sufficient frameworks for complete synchronization. For a fixed $\delta \in (0,1)$, our first framework is given as follows.

\begin{itemize}
\item
$(\mathcal{F}_A1)$: System parameters $m$, $\gamma$, $\kappa_0$ and $\delta$ satisfy
\[
\gamma^2 -16m\kappa_0\delta > 0, \quad m, \gamma, \kappa_0 > 0, \quad \kappa_1 \geq 0.
\]
\item
$(\mathcal{F}_A2)$: Initial data satisfy
\[
{\mathcal G}(0) < \frac{8\kappa_1 +16mM_1^2}{4\kappa_0\delta} < \frac{(1-\delta)^2}{N}, \quad \dot{{\mathcal G}}(0) +\nu_1 {\mathcal G}(0) < \frac{\nu_1(8\kappa_1 +16mM_1^2)}{4\kappa_0\delta}.
\]
\end{itemize}

And also, our second framework is given as follows. \newline

\begin{itemize}
\item
$(\mathcal{F}_B1)$: System parameters $m$, $\gamma$, $\kappa_0$ and $\delta$ satisfy
\[
\gamma^2 -16m\kappa_0\delta < 0, \quad m, \gamma >0, \quad \kappa_1 \geq 0.
\]
\item
$(\mathcal{F}_B2)$: Initial data satisfy
\[
\mathcal G(0) < \frac{4m}{\gamma^2}(8\kappa_1 +16mM_1^2) < \frac{(1-\delta)^2}{N}, \quad \dot{{\mathcal G}}(0) +\frac{\gamma}{2m} {\mathcal G}(0) < \frac{2}{\gamma}(8\kappa_1 +16mM_1^2).
\]
\end{itemize}

Our first main result is concerned with the complete aggregation of a homogeneous ensemble.
\begin{theorem} \label{T4.1}
Suppose that the sufficient frameworks $(\mathcal F_A1)$-$(\mathcal F_A2)$ or $(\mathcal F_B1)$-$(\mathcal F_B2)$ hold. Moreover, assume that initial data and natural frequency satisfy
\[
\|z_j^{in}\| = 1, \quad \Omega_j = 0, \quad j=1,\cdots,N.
\]
Let $\{ z_j \}$ be the global solution of \eqref{A-1}. Then, we have
\begin{equation*}
\lim_{t \to \infty} {\mathcal G}(t) = 0, \quad \mbox{i.e.,} \quad  \lim_{t \to \infty} h_{ij}(t) = 1, \quad \forall~1 \leq i,j \leq N.
\end{equation*}
\end{theorem}
\begin{proof} Although the detailed proof can be found in Section \ref{sec:5}, we briefly sketch some ingredients for reader's convenience. Since
\[ g_{ij} +g_{ji} = 2 -\langle z_i, z_j \rangle -\langle z_j, z_i \rangle = \|z_i - z_j \|^2, \]  
one has 
\[ \lim_{t \to \infty} {\mathcal G}(t) = 0 \quad \Longrightarrow \quad  \lim_{t \to \infty} {\mathcal D}(Z(t)) = 0. \]
Thus, it suffices to verify 
\begin{equation} \label{D-2-1}
   \lim_{t\to\infty} {\mathcal G}(t) = 0. 
 \end{equation}  
By straightforward calculation to be performed in next section, we can derive second-order differential inequality for ${\mathcal G}$ in \eqref{C-1}:
\begin{equation*}
m \ddot{{\mathcal G}} +\gamma \dot{{\mathcal G}} +4\kappa_0 {\mathcal G}  \leq 4\kappa_0 \sqrt{N} {\mathcal G}^\frac{3}{2} +2\kappa_1 {\mathcal R}_2(Z) +16m {\mathcal R}_1(\dot Z).
\end{equation*}
Next, we use sufficient frameworks to show the uniform boundedness of $\mathcal G$, which yields
\begin{equation}\label{D-3}
m \ddot{{\mathcal G}} +\gamma \dot{{\mathcal G}} +4\delta\kappa_0 {\mathcal G}  \leq 2\kappa_1 {\mathcal R}_2(Z) +16m {\mathcal R}_1(\dot Z).
\end{equation}
Then, we use the relations in Proposition \ref{P5.1} and \eqref{F-13}:
\begin{equation} \label{D-4}
\lim_{t \to \infty} {\mathcal R}_2(Z(t)) = 0, \quad \lim_{t \to \infty} {\mathcal R}_1({\dot Z}(t)) = 0,
 \end{equation}
and the second-order Gronwall's inequality \eqref{D-3} together with \eqref{D-4} to derive \eqref{D-2-1}. 
\end{proof}
\subsection{Practical aggregation} \label{sec:4.2} 
In this subsection, we first list a framework $(\mathcal{F}_C)$ formulated in terms of system parameters and initial data for a practical synchronization.  \newline

First, we introduce several notation:
\begin{align*}
\begin{aligned}
& {\mathcal R}_3(V) := \max_{j} \| v_j \|, \quad \Omega^{\infty} :=  \max_{j} \|\Omega_j \|_F, \\
& U(m, \Omega^{\infty}, \kappa_0, \kappa_1, \gamma):=4\Omega^\infty +8\kappa_1 +\frac{16m}{\gamma^2}\big[ \Omega^\infty +2(\kappa_0 +\kappa_1) \big]^2.
\end{aligned}
\end{align*} 

\vspace{0.2cm}

Now, we set up a sufficient framework for practical aggregation. For a fixed $\delta \in (0,1)$, our framework is given as follows. \newline
\begin{itemize}
\item
$(\mathcal{F}_C1)$: System parameters $m$, $\gamma$, $\kappa_0$ and $\delta$ satisfy
\[
\gamma^2 -16m\kappa_0\delta > 0, \quad m, \gamma > 0, \quad \kappa_1 \geq 0.
\]
\item
$(\mathcal{F}_C2)$: Initial data satisfy
\[
\begin{cases}
\displaystyle {\mathcal R}_3(V^{in}) < \frac{2}{\gamma}(\kappa_0 +\kappa_1), \quad {\mathcal G}(0) < \frac{1}{4\kappa_0\delta} ~U(m, \Omega^{\infty}, \kappa_0, \kappa_1, \gamma) < \frac{(1-\delta)^2}{N}, \\
\displaystyle \dot{{\mathcal G}}(0) +\nu_1 {\mathcal G}(0) < \frac{\nu_1}{4\kappa_0\delta}~U(m, \Omega^{\infty}, \kappa_0, \kappa_1, \gamma).
\end{cases}
\]
\end{itemize}

\vspace{0.2cm}

\noindent Since practical aggregation is discussed with sufficiently large $\kappa_0$, there is no state about $\kappa_0$ in the framework $\mathcal F_C$. 

Under the above framework, our second result deals with the emergence of practical aggregation for a heterogeneous ensemble.
\begin{theorem} \label{T4.2}
Suppose that the sufficient framework $(\mathcal{F}_C1)$-$({\mathcal F}_C2)$ holds, and let $\{ z_j \}$ be the solution of \eqref{A-1} with $\| z_j^{in} \| = 1$, $j = 1, \cdots, N$. Then, we have a practical aggregation:
\[
\lim_{\kappa_0 \to \infty}\limsup_{t \to \infty} {\mathcal G}(t) = 0.
\]
\end{theorem}
\begin{proof} We briefly sketch a key idea. Detailed argument can be found in Section \ref{sec:6}. In the course of proof, we will derive the following differential inequality:
\[
m\ddot{{\mathcal G}} +\gamma \dot{{\mathcal G}} +4\kappa_0 \delta {\mathcal G} \leq 4\Omega^\infty +8\kappa_1 +\frac{16m}{\gamma^2}\big[ \Omega^\infty +2(\kappa_0 +\kappa_1) \big]^2, \quad \forall ~t \in (0, T_*).
\]
Then, this yields
\[
{\mathcal G}(t) < \frac{\Omega^\infty +2\kappa_1}{\kappa_0\delta} +\frac{4m}{\gamma^2\kappa_0\delta}\big[ \Omega^\infty +2(\kappa_0 +\kappa_1) \big]^2, \quad \forall ~t>0.
\]
For a sufficiently large $\kappa_0 \geq \max \big\{ \Omega^\infty, 2\kappa_1 \big\}$ and a suitable ansatz  for $m$:
\[ m = \frac{m_0}{\kappa^{1 + \eta}}, \quad \eta > 0, \quad m_0 > 0, \] 
one has 
\[
\limsup_{t \to \infty} {\mathcal G}(t) <  \frac{\Omega^\infty +2\kappa_1}{\kappa_0\delta} +\frac{64}{\gamma^2\delta} \cdot \frac{m_0}{\kappa_0^\eta}.
\]
This implies the desired result.
\end{proof}

%%%%%%%%%%%%%%%%%%%%%%%%
%
%
%
%%%%%%%%%%%%%%%%%%%%%%%%%%
\section{Emergence of complete aggregation}\label{sec:5}
\setcounter{equation}{0}
In this section, we provide estimates on the complete aggregation to the second-order LHS model with inertia for a homogeneous ensemble: 
\[ \Omega_j = \Omega, \quad \|z_j \| = 1, \quad j = 1, \cdots, N. \]
Furthermore, by Lemma \ref{L2.2}, without loss of generality, we may assume $\Omega = 0$. In this situation, $z_j$ satisfies 
\begin{equation}\label{F-0}
\begin{cases}
\displaystyle m \ddot{z}_j = -\gamma \dot{z}_j + \kappa_0 \big(z_c -\left\langle{z_c, z_j }\right\rangle z_j \big) +\kappa_1 \big( \left\langle{ z_j, z_c }\right\rangle -\left\langle{ z_c, z_j }\right\rangle \big) z_j - m \| \dot{z}_j \|^2 z_j,  \\
\displaystyle z_j(0) = z_j^{in}, \quad \dot{z}_j(0) = \dot{z}_j^{in}, \quad \langle{z_j^{in}, \dot{z}_j^{in}}\rangle+\langle{\dot{z}_j^{in}, z_j^{in}}\rangle = 0, \quad j = 1, \cdots, N.
\end{cases}
\end{equation}
Since the proof of Theorem \ref{T4.1} is very lengthy, we briefly delineate a proof strategy in several steps. Recall that our main purpose in this section is to derive a sufficient frameworks (setting) leading to the complete aggregation:
\begin{equation} \label{F-0-0}
 \lim_{t \to \infty} \langle z_i, z_j \rangle = 1, \quad \mbox{i.e.,} \quad \lim_{t \to \infty} {\mathcal D}(Z(t)) = 0.
\end{equation}

\vspace{0.2cm}

\begin{itemize}
\item
Step A:~We introduce an energy functional ${\mathcal E}$ and via a time-decay estimate of it, we show that 
\[ \lim_{t \to \infty} \|{\dot z}_j(t) \| = 0, \quad \lim_{t \to \infty}   | \langle{ z_c, z_{j} \rangle} -\langle{ z_{j}, z_c \rangle} | = 0, \quad j = 1, \cdots, N.   \]
See Proposition \ref{P5.1} for details. 

\vspace{0.2cm}

\item
Step B:~We derive a second-order differential inequality for ${\mathcal G}$:
\begin{equation} \label{F-0-1}
m \ddot{{\mathcal G}} +\gamma \dot{{\mathcal G}} +4\kappa_0 {\mathcal G}  \leq 4\kappa_0 \sqrt{N} {\mathcal G}^\frac{3}{2} +  f(t), \quad f(t) \to 0 \quad \mbox{as $t \to \infty$}. 
\end{equation}

\item
Step C:~We use a second-order Gronwall's lemma (Lemma \ref{L5.4}) and the result of Step A to derive a zero convergence of ${\mathcal G}$:
\[ \lim_{t \to \infty} {\mathcal G}(t) = 0, \]
which implies \eqref{F-0-0}.
\end{itemize} 

In the following two subsections, we perform the above three steps one by one.

\vspace{0.5cm}

\subsection{Zero convergence of energy functional} \label{sec:5.1} For a solution $\{ z_j \}$ to \eqref{A-1},  we define an energy functional:
\begin{align*}
\mathcal E & := \frac{1}{N}\sum_{j=1}^N\left(m\|\dot{z}_j\|^2-m\frac{\kappa_0+2\kappa_1}{2\left(\kappa_0+\kappa_1\right)}\left|\langle{z_j,\dot{z}_j}\rangle\right|^2+\kappa_0\|z_c-z_j\|^2\right) \\
& = \frac{1}{N}\sum_{j=1}^Nm\left(\|\dot{z}_j\|^2-\frac{\kappa_0+2\kappa_1}{2\left(\kappa_0+\kappa_1\right)}\left|\langle{z_j,\dot{z}_j}\rangle\right|^2\right)+\kappa_0\left(1-\|z_c\|^2\right).
\end{align*}
In the following lemma, we check the following two properties of $\mathcal E$:
\begin{enumerate}[\hspace{.5cm}1.]
\item ~${\mathcal E} \geq 0$, \vspace{.1cm}
\item ~${\mathcal E} = 0 \quad \iff \quad \| z_c \|  = 1 \quad \mbox{and} \quad \| {\dot z}_j \|  = 0, \quad j = 1, \cdots, N$.
\end{enumerate}
\begin{lemma} \label{L5.1}
Suppose the coupling strengths $\kappa_0$ and $\kappa_1$ satisfy 
\[ \kappa_0>0 \quad \mbox{and} \quad \kappa_1\geq 0, \]
and let $\{ z_j \}$ be a solution to \eqref{F-0}. Then the following assertions hold.
\begin{enumerate}
\item
The energy functional ${\mathcal E}$ is nonnegative:
\[
{\mathcal E}(t) \geq 0, \quad \forall~t \geq 0.
\]
\item
The energy functional $\mathcal{E}$ is zero if and only if 
\[ \|z_c\|=1  \quad \mbox{and} \quad \|\dot{z}_j\| = 0, \quad j = 1, \cdots, N. \]
\end{enumerate}
\end{lemma}
\begin{proof}
\noindent (i)~The first assertion follows from 
\[ 0 < \frac{\kappa_0+2\kappa_1}{2\left(\kappa_0+\kappa_1\right)} < 1, \quad  \left|\langle{z_j,\dot{z}_j}\rangle\right|^2 \leq \| z_j \|^2 \cdot \|\dot{z}_j\|^2, \quad \kappa_0 > 0, \quad \|z_c \| \leq 1.    \]
\noindent (ii)~Note that 
\begin{align*}
\begin{aligned}
{\mathcal E} = 0 \quad & \iff \quad \|\dot{z}_j\|^2-\frac{\kappa_0+2\kappa_1}{2\left(\kappa_0+\kappa_1\right)}\left|\langle{z_j,\dot{z}_j}\rangle\right|^2 = 0, \quad \forall~i = 1, \cdots, N, 
\quad 1-\|z_c\|^2 = 0 \\
& \iff \quad {\dot z}_j = 0, \quad \forall~j = 1, \cdots, N, \quad \mbox{and} \quad  \| z_c \| = 1.
\end{aligned}
\end{align*}
\end{proof}

Next, we study a nonincreasing property of ${\mathcal E}$ along system \eqref{A-1}.
\begin{lemma} \label{L5.2}
Let $\{z_j\}$ be the solution of \eqref{F-0}. Then, we have
\begin{equation} \label{F-1}
\frac{d\mathcal{E}}{dt} = -\frac{2\gamma}{N}\sum_{j=1}^N\left(\| \dot{z}_j  \|^2-\frac{\kappa_0+2\kappa_1}{2\left(\kappa_0+\kappa_1\right)}\left|\langle{z_j, \dot{z}_j}\rangle\right|^2\right)\leq0, \quad t > 0.
\end{equation}
\end{lemma}
\begin{proof}
First, we use the relation $\|z_j \|^2 = 1$ to see 
\begin{equation} \label{F-2}
\langle{z_j,\dot{z}_j}\rangle+\langle{\dot{z}_j,z_j}\rangle = 0.
\end{equation}
We use \eqref{F-2} to obtain
\begin{align}
\begin{aligned} \label{F-3}
&m\frac{d}{dt}\|\dot{z}_j\|^2 = \langle{\dot{z}_j,m\ddot{z}_j}\rangle+\langle{m\ddot{z}_j,\dot{z}_j}\rangle \\
& \hspace{0.5cm} = -2\gamma\|\dot{z}_j\|^2 +\big[ \kappa_0\left(\langle{\dot{z}_j, z_c}\rangle-\langle{z_c, z_j}\rangle\langle{\dot{z}_j, z_j}\rangle\right) + (c.c) \big] \\
& \hspace{0.7cm} +\big[ \kappa_1\left(\langle{z_j, z_c}\rangle-\langle{z_c, z_j}\rangle\right)\langle{\dot{z}_j, z_j}\rangle+(c.c) \big] -\big[ m\|\dot{z}_j\|^2\langle{\dot{z}_j, z_j}\rangle+(c.c) \big] \\
& \hspace{.5cm} = -2\gamma\|\dot{z}_j\|^2 +\kappa_0\left( \langle{z_c, \dot{z}_j}\rangle +\langle{\dot{z}_j, z_c}\rangle \right) -\kappa_0 \left( \langle{z_c, z_j}\rangle -\langle z_j, z_c \rangle \right) \langle{\dot{z}_j, z_j}\rangle  \\
& \hspace{.7cm} +2\kappa_1\left(\langle{z_j, z_c}\rangle-\langle{z_c, z_j}\rangle\right)\langle{\dot{z}_j, z_j}\rangle \\
& \hspace{0.5cm} = -2\gamma\|\dot{z}_j\|^2+\kappa_0\left(\langle{z_c,\dot{z}_j}\rangle+\langle{\dot{z}_j, z_c}\rangle\right) + \left(\kappa_0+2\kappa_1\right)\left(\langle{z_c, z_j}\rangle-\langle{z_j, z_c}\rangle\right)\langle{z_j, \dot{z}_j}\rangle,
\end{aligned}
\end{align}
where $(c.c)$ means the complex conjugate of the previous term. \newline

We take summation \eqref{F-3} over $j$ and divide by $N$ to obtain
\begin{align*}
\begin{aligned}
& \frac{d}{dt}\left(\frac{1}{N}\sum_{j=1}^Nm\|\dot{z}_j\|^2\right) \\
& \hspace{.2cm} = -\frac{2\gamma}{m}\left(\frac{1}{N}\sum_{j=1}^Nm\|\dot{z}_j\|^2\right) + \kappa_0\frac{d}{dt}\|z_c\|^2 + \frac{\kappa_0+2\kappa_1}{N}\sum_{j=1}^N\Big(\langle{z_c, z_j}\rangle-\langle{z_j, z_c}\rangle \Big)\langle{z_j, \dot{z}_j}\rangle,
\end{aligned}
\end{align*}
or equivalently
\begin{equation}\label{F-4}
\begin{split}
& \frac{d}{dt}\left[\frac{1}{N}\sum_{j=1}^Nm\|\dot{z}_j\|^2 + \kappa_0\left(1-\|z_c\|^2\right)\right] \\
& \hspace{.5cm} = -\frac{2\gamma}{m}\left(\frac{1}{N}\sum_{j=1}^Nm\|\dot{z}_j\|^2\right) + \frac{\kappa_0+2\kappa_1}{N}\sum_{j=1}^N\left(\langle{z_c, z_j}\rangle-\langle{z_j, z_c}\rangle\right)\langle{z_j, \dot{z}_j}\rangle.
\end{split}
\end{equation}
On the other hand, one has 
\begin{align}
\begin{aligned} \label{F-5}
m\frac{d}{dt}\left|\langle{z_j, \dot{z}_j}\rangle\right|^2 & = m\frac{d}{dt}\left(\langle{z_j, \dot{z}_j}\rangle\langle{\dot{z}_j, z_j}\rangle\right) \\
& = m\left[\left(\|\dot{z}_j\|^2+\langle{z_j, \ddot{z}_j}\rangle\right)\langle{\dot{z}_j, z_j}\rangle+\langle{z_j, \dot{z}_j}\rangle\left(\|\dot{z}_j\|^2+\langle{\ddot{z}_j, z_j}\rangle\right)\right].
\end{aligned}
\end{align} 
Then, we use \eqref{F-4}, \eqref{F-5} and the following relation:
\[
m\left(\langle{z_j, \ddot{z}_j}\rangle+\|\dot{z}_j\|^2\right) = -\gamma\langle{z_j, \dot{z}_j}\rangle+\left(\kappa_0+\kappa_1\right)\left(\langle{z_j, z_c}\rangle-\langle{z_c, z_j}\rangle\right)
\]
to obtain
\[
m\frac{d}{dt}\left|\langle{z_j, \dot{z}_j}\rangle\right|^2 = -2\gamma\left|\langle{z_j, \dot{z}_j}\rangle\right|^2 + 2\left(\kappa_0+\kappa_1\right)\left(\langle{z_c, z_j}\rangle-\langle{z_j, z_c}\rangle\right)\langle{z_j, \dot{z}_j}\rangle,
\]
or equivalently
\begin{equation} \label{F-6}
\left(\langle{z_c, z_j}\rangle-\langle{z_j, z_c}\rangle\right)\langle{z_j, \dot{z}_j}\rangle = \frac{m}{2(\kappa_0  + \kappa_1)} \frac{d}{dt} |\langle z_j, {\dot z}_j \rangle|^2 + \frac{\gamma}{\kappa_0 + \kappa_1} |\langle z_j, {\dot z}_j \rangle|^2. 
\end{equation}
Finally, we combine \eqref{F-4} and \eqref{F-6} to get the desired result.
\end{proof}
\begin{remark} Note that the estimate \eqref{F-1} can be rewritten as 
\begin{equation} \label{F-7}
\frac{d{\mathcal E}}{dt} = -\frac{2\gamma}{m} {\mathcal E} + \frac{2\kappa_0 \gamma}{m} (1 - \| z_c \|^2), \quad \forall ~t > 0.
\end{equation}
 \end{remark}
As a corollary of Lemma \ref{L5.1} and Lemma \ref{L5.2}, we obtain the following result.
\begin{corollary} \label{C5.1}
Suppose system parameters satisfy 
\[ m > 0, \quad \gamma > 0, \quad \kappa_0>0 \quad \mbox{and} \quad \kappa_1\geq 0, \]
and $\{z_j\}$ be the solution of \eqref{F-0}. Then, we have the following estimates:
\begin{align*}
\begin{aligned}
& (i)~\exists~{\mathcal E}_\infty := \lim_{t \to \infty} {\mathcal E}(t). \\
& (ii)~\max_{1 \leq j \leq N} \|  \dot{z}_j \| \leq \max \bigg\{ \|w_1^{in} \|, \cdots, \|w_N^{in} \|, \frac{2}{\gamma}(\kappa_0 +\kappa_1) \bigg\} =: M_1.  \\
& (iii)~\lim_{t  \to \infty} \|z_c \| = 1 \quad \Longrightarrow \quad \lim_{t \to \infty} {\mathcal E}(t) = 0.
\end{aligned}
\end{align*}
\end{corollary}
\begin{proof}
\noindent (i)~Since 
\[
\mathcal{E}\left( t \right) \geq 0, \quad \dot{\mathcal{E}}\left(t\right) \leq 0, \quad \forall ~t > 0,
\]
${\mathcal E}$ converges as $t \to \infty$.  

\vspace{0.5cm}

\noindent (ii)~It follows from \eqref{F-3} that if $ \kappa_0>0$ and $\kappa_1\geq0$, we have
\begin{align*}
m\frac{d}{dt}\|\dot{z}_j\|^2 \leq -2\gamma\|\dot{z}_j\|^2 +4( \kappa_0 +\kappa_1) \| \dot{z}_j \|,
\end{align*}
or equivalently,
\begin{align*}
\frac{d}{dt}\|\dot{z}_j\| \leq -\frac{\gamma}{m} \| \dot{z}_j \| +\frac{2}{m}( \kappa_0 +\kappa_1).
\end{align*}
This implies
\begin{align*}
\| \dot{z}_j \| \leq \bigg( \|w_j^{in} \| -\frac{2}{\gamma}(\kappa_0 +\kappa_1) \bigg) e^{-\frac{\gamma}{m} t} +\frac{2}{\gamma}(\kappa_0 +\kappa_1), \quad \forall ~t > 0.
\end{align*}
Hence, for all $j$, we have
\begin{align*}
\| \dot{z}_j \| \leq \max \bigg\{ \|w_1^{in} \|, \cdots, \|w_N^{in} \|, \frac{2}{\gamma}(\kappa_0 +\kappa_1) \bigg\}.
\end{align*}

\vspace{0.2cm}

\noindent (iii)~It follows from \eqref{F-7} that
\begin{equation}\label{F-8}
{\mathcal E}(t) = {\mathcal E}(0) e^{-\frac{2\gamma t}{m}} + \frac{2\gamma \kappa_0}{m} \int_0^t e^{-\frac{2\gamma}{m}(t-s)} (1 - \|z_c(s) \|^2) ds, \quad \forall ~t \geq 0.
\end{equation}
Suppose that 
\[ \lim_{t \to \infty} \| z_c(t) \| = 1. \]
Then, for any positive small $\varepsilon$, there exists a positive time $T = T(\varepsilon) > 0$ such that
\begin{equation*}
1-\| z_c (t) \|^2 < \varepsilon, \quad \forall~t > T(\varepsilon).
\end{equation*}
Then, \eqref{F-8} becomes
\begin{align*}
{\mathcal E}(t) & = {\mathcal E}(0) e^{-\frac{2\gamma t}{m}} + \frac{2\gamma \kappa_0}{m} \int_0^{T(\varepsilon)} e^{-\frac{2\gamma}{m}(t-s)} (1 - \|z_c(s) \|^2) ds \\
& \hspace{.2cm} +\frac{2\gamma \kappa_0}{m} \int_{T(\varepsilon)}^t e^{-\frac{2\gamma}{m}(t-s)} (1 - \|z_c(s) \|^2) ds \\
& \leq {\mathcal E}(0) e^{-\frac{2\gamma t}{m}} + \frac{2\gamma \kappa_0}{m} \int_0^{T(\varepsilon)} e^{-\frac{2\gamma}{m}(t-s)} ds +\frac{2\gamma \kappa_0\varepsilon}{m} \int_{T(\varepsilon)}^t e^{-\frac{2\gamma}{m}(t-s)} ds \\
& = {\mathcal E}(0) e^{-\frac{2\gamma t}{m}} +\kappa_0e^{-\frac{2\gamma}{m}t} \Big( e^{\frac{2\gamma}{m}T(\varepsilon)} -1 \Big) +\kappa_0\varepsilon \Big( 1 -e^{-\frac{2\gamma}{m}(t-T(\varepsilon))} \Big), \quad \forall ~t > T(\varepsilon).
\end{align*}
This implies that for $t \gg 1$,
\[ {\mathcal E}(t) \leq 2 \kappa_0 \varepsilon. \]
Since $\varepsilon$ is arbitrary, we have desired zero convergence of ${\mathcal E}$. 
\end{proof}
\begin{remark}\label{R5.2}
Uniform boundedness of $\dot z_j$ provides us the uniform boundedness of $\ddot z_j$ since
\begin{align*}
m \|\ddot z_j\| &= \big\| -\gamma \dot{z}_j + \kappa_0 \big(z_c -\left\langle{z_c, z_j }\right\rangle z_j \big) +\kappa_1 \big( \left\langle{ z_j, z_c }\right\rangle -\left\langle{ z_c, z_j }\right\rangle \big) z_j - m \| \dot{z}_j \|^2 z_j \big\| \\
& \leq \gamma\|\dot z_j\|  +2(\kappa_0 +\kappa_1) +m\|\dot z_j\|^2.
\end{align*}
Similarly, one can also get the uniform boundedness of $\frac{d^3z_j}{dt^3}$.
\end{remark}
In next proposition, we study zero convergence of ${\dot z}_j$.
\begin{proposition}\label{P5.1}
Suppose system parameters and initial data satisfy
\[ m>0, \quad \gamma>0, \quad \kappa_0>0,\quad \kappa_1\geq0,  \quad \|z_j^{in} \| = 1 \quad \mbox{for all $j=1,\cdots,N$}, 
\]
and let $\{z_j\}$ be the solution of \eqref{F-0}. Then, we have
\[ \lim_{t\to\infty}\|\dot{z}_j(t)\| = 0, \quad j = 1,\cdots,N. \]
\end{proposition}
\begin{proof}
We integrate \eqref{F-1} to find
\[
\mathcal{E}\left( t \right)+\frac{2\gamma}{N}\sum_{j=1}^N\int_{0}^t\left(\|\dot{z}_j\left(s\right)\|^2-\frac{\kappa_0+2\kappa_1}{2\left(\kappa_0+\kappa_1\right)}\left|\langle{z_j\left(s\right),\dot{z}_j\left(s\right)}\rangle\right|^2\right)ds = \mathcal{E}^{in} < \infty.
\]
This yields
\[
\int_{0}^\infty\left(\|\dot{z}_j\left(t\right)\|^2-\frac{\kappa_0+2\kappa_1}{2\left(\kappa_0+\kappa_1\right)}\left|\langle{z_j\left(t\right),\dot{z}_j\left(t\right)}\rangle\right|^2\right)dt < \infty.
\]
By straightforward calculation, one can show that time-derivative of the integrand is uniformly bounded because $z_j, {\dot z}_j$ and ${\ddot z}_j$ are bounded, i.e., the integrand is uniformly continuous. Hence, we can apply Barbalat's lamma to get
\begin{equation} \label{F-10}
\lim_{t\to\infty}\left(\|\dot{z}_j\left(t\right)\|^2-\frac{\kappa_0+2\kappa_1}{2\left(\kappa_0+\kappa_1\right)}\left|\langle{z_j\left(t\right),\dot{z}_j\left(t\right)}\rangle\right|^2\right) = 0.
\end{equation}
On the other hand, note that 
\begin{align}\label{F-11}
0\leq\frac{\kappa_0}{2\left(\kappa_0+\kappa_1\right)}\|\dot{z}_j\left(t\right)\|^2\leq\|\dot{z}_j\left(t\right)\|^2-\frac{\kappa_0+2\kappa_1}{2\left(\kappa_0+\kappa_1\right)}\left|\langle{z_j\left(t\right),\dot{z}_j\left(t\right)}\rangle\right|^2.
\end{align}
Finally, we combine estimates \eqref{F-10} and \eqref{F-11} to derive the desired zero convergence of ${\dot z}_j$. 
\end{proof}

\vspace{0.2cm}

Next, we study all possible equilibria in terms of order parameter $\rho^{\infty}$.  
\begin{corollary} \label{C5.2}
Under the same assumptions of Theorem \ref{T4.1}, let $\{z_j\}$ be the solution of \eqref{A-1} and let $\rho^{\infty}$ be an asymptotic order parameter defined by \eqref{C-10-1}. Then, the following trichotomy holds:
\[ \rho^\infty = 0, \quad \rho^\infty = 1, \]
or there exists integer $n$ such that
\[ \rho^\infty = \frac{N-2n}{N}, \quad 1 \leq n < \left\lfloor\frac{N}{2}\right\rfloor \].
\end{corollary}
\begin{proof}
For the case $\rho^\infty = 0$ or $1$, we are done. Hence, we consider only the case:
\[  \rho^{\infty} \in (0, 1). \] 
In Remark \ref{R5.2}, we noticed that $\frac{d^3 z_j}{dt^3}$ is uniformly bounded. Hence, we can apply Barbalat's Lemma \cite{Ba} and Proposition \ref{P5.1} to have
\[
\lim_{t\to\infty} \ddot z_j (t) = 0.
\]
Then, \eqref{F-0}$_1$ becomes
\begin{align}\label{F-12}
\lim_{t\to\infty}\left(\kappa_0(z_c-\langle{z_c, z_j}\rangle z_j)+\kappa_1(\langle{z_j, z_c}\rangle-\langle{z_c, z_j}\rangle)z_j\right) = 0, \quad j = 1,\cdots, N.
\end{align}
Since $z_j$ is bounded, we can take $\langle{z_j,\cdot\ }\rangle$ to get
\begin{align}\label{F-13}
\lim_{t\to\infty}\left(\langle{z_j, z_c}\rangle-\langle{z_c, z_j}\rangle\right) = 0, \quad j = 1,\cdots, N.
\end{align}
We combine \eqref{F-12}, \eqref{F-13} with the fact that $\|z_j\| = 1$ to obtain
\begin{equation*}
\lim_{t\to\infty}\Big( z_c-\langle{z_c,z_j}\rangle z_j \Big) = 0, \quad j = 1,\cdots, N.
\end{equation*}
Again, since $z_c$ is bounded, we can take $\langle{z_c,\cdot\ }\rangle$ to get
\begin{align*}
\lim_{t\to\infty}\left( \rho^2-\langle{z_c,z_j}\rangle^2\right) = 0, \quad j = 1,\cdots, N,
\end{align*}
or equivalently,
\begin{equation} \label{F-15}
\lim_{t\to\infty}\langle{z_c,z_j}\rangle = \delta_j\rho^\infty, \quad \delta_j \in \{1,-1\}, \quad j = 1,\cdots, N.
\end{equation}
Then, we sum \eqref{F-15} over $j$ and divide by $N$ to obtain
\begin{align*}
\big( \rho^\infty \big)^2 = \lim_{t\to\infty} \|z_c\|^2 = \frac{\rho^\infty}{N}\sum_{j=1}^N\delta_j.
\end{align*}
This implies
\[
\rho^\infty = \frac{1}{N}\sum_{j=1}^N\delta_j.
\]
Since $\rho^\infty > 0$, there must be an integer $n$ such that
\begin{align*}
n = |\{ j: \delta_j = -1, \quad j = 1, \cdots, N \}|, \quad 1 \leq n < \left\lfloor \frac{N}{2} \right\rfloor,
\end{align*}
which guarantees our desired result.
\end{proof}

\begin{remark}
In the results of the above lemma, we call first two cases by \begin{enumerate}[1.]
\item ~$\rho = 0 \iff$ incoherence state,
\item ~$\rho = 1 \iff$ complete aggregation.
\end{enumerate}
On the other hand, the remaining case is called bi-polar state, which is defined as follows:
\begin{align*}
\lim_{t\to\infty} \mbox{\normalfont dist} \big( S, \{z_j\} \big) = 0, \quad S:= \big\{ \{p_j\} \in (\bbh\bbs^d)^N: p_j = \pm a, \quad \exists ~a \in \bbh\bbs^d \big\}.
\end{align*}
\end{remark}
%Moreover, we have
%\begin{align*}
%& \lim_{t\to\infty}\ddot{z}_j(t) \\
%& \hspace{.2cm} = \lim_{t\to\infty}\left(-\frac{\gamma}{m}\dot{z}_j+\frac{\kappa_0}{m}(\langle{z_j, z_j}\rangle z_c-\langle{z_c, z_j}\rangle z_j)+\frac{\kappa_1}{m}(\langle{z_j, z_c}\rangle-\langle{z_c, z_j}\rangle)z_j-\|\dot{z}_j\|^2z_j\right) = 0.
%\end{align*}

\vspace{.1cm}

\subsection{Proof of Theorem \ref{T4.1}} \label{sec:5.2}
In this subsection, we provide a proof of our first main result by analyzing asymptotic behaviors of angle parameter ${\mathcal G}$ and diameter functional ${\mathcal R}_i(Z)$. First, we recall two-point correlation functions $h_{ij}$ and $g_{ij}$:
\[ h_{ij} = \langle z_i, z_j \rangle, \quad g_{ij} = 1-h_{ij}, \quad \forall ~i, j = 1, \cdots, N. \] 
Next,  we derive an evolution equation for $|g_{ij}|^2$. 
\begin{lemma} \label{L5.3}
Let $\{ z_j \}$ be a solution of \eqref{F-0} with $\| z_j^{in} \| = 1$, $j = 1, \cdots, N$. Then, $| g_{ij} |^2$ satisfies
\begin{align}
\begin{aligned} \label{F-16}
& m \frac{d^2}{dt^2} | g_{ij} |^2 +\gamma \frac{d}{dt} | g_{ij} |^2 +2\left[ 2\kappa_0 +m ( \| \dot{z}_i \|^2 +\| \dot{z}_j \|^2 ) \right] | g_{ij} |^2 \\
& \hspace{0.2cm} = \frac{\kappa_0}{N} \sum_{k=1}^N ( g_{ik} +g_{ki} +g_{kj} +g_{jk} ) | g_{ij} |^2 +\kappa_1 ( \langle{ z_c, z_j \rangle} -\langle{ z_j, z_c \rangle} ) ( \langle{ z_i, z_j \rangle} -\langle{ z_j, z_i \rangle}) \\
& \hspace{.5cm} +\kappa_1 ( \langle{ z_i, z_c \rangle} -\langle{ z_c, z_i \rangle}) ( \langle{ z_i, z_j \rangle} -\langle{ z_j, z_i \rangle}) +2m\dot{g}_{ij}\dot{g}_{ji} \\
& \hspace{.5cm} +m ( \| \dot{z}_j \|^2 -2  \langle{ \dot{z}_i, \dot{z}_j \rangle} +\| \dot{z}_i \|^2 ) g_{ji} +m ( \| \dot{z}_i \|^2 -2  \langle{ \dot{z}_j, \dot{z}_i \rangle} +\| \dot{z}_j \|^2 ) g_{ij}.
\end{aligned}
\end{align}
\end{lemma}
\begin{proof} Recall that $z_j$ satisfies 
\begin{equation} \label{F-17}
 m \ddot{z}_j = -\gamma \dot{z}_j + \kappa_0 \big(z_c -\left\langle{z_c, z_j }\right\rangle z_j \big) +\kappa_1 \big( \left\langle{ z_j, z_c }\right\rangle -\left\langle{ z_c, z_j }\right\rangle \big) z_j - m \| \dot{z}_j \|^2 z_j.
\end{equation}
We use $\overline{g_{ij}} = g_{ji}$ to find
\[
m\frac{d^2}{dt^2} |g_{ij}|^2 = m\frac{d^2}{dt^2} (g_{ij}g_{ji}) = m\ddot{g}_{ij}g_{ji} +2m\dot{g}_{ij}\dot{g}_{ji} +m\ddot{g}_{ji}g_{ij}.
\]
On the other hand, it follows from $g_{ij} = 1-\langle z_i, z_j \rangle$ that
\begin{equation} \label{F-18}
m\ddot{g}_{ij} = -\langle{ m\ddot{z}_i, z_j \rangle} -2m \langle{ \dot{z}_i, \dot{z}_j \rangle} -\langle z_i, m\ddot{z}_j \rangle.
\end{equation}
Then, we use \eqref{F-17} to get 
\begin{align}
\begin{aligned} \label{F-19}
\langle{ z_i, m\ddot{z}_j \rangle} &= -\gamma \langle{ z_i, \dot{z}_j \rangle} +\kappa_0 \langle{ z_i, z_c \rangle} -\kappa_0 \langle{ z_c, z_j \rangle} \langle{ z_i, z_j \rangle} \\
& \hspace{1cm} +\kappa_1 ( \langle{ z_j, z_c \rangle} -\langle{ z_c, z_j \rangle} ) \langle{ z_i, z_j \rangle} -m \| \dot{z}_j \|^2 \langle{ z_i, z_j \rangle}.
\end{aligned}
\end{align}
We combine \eqref{F-18}  and \eqref{F-19} to obtain
\begin{align*}
m \ddot{g}_{ij} & = \gamma ( \langle{ z_i, \dot{z}_j \rangle} +\langle{ \dot{z}_i, z_j \rangle} ) -\kappa_0 ( \langle{ z_i, z_c \rangle} +\langle{ z_c, z_j \rangle} ) +\kappa_0 ( \langle{ z_c, z_j \rangle} +\langle{ z_i, z_c \rangle} ) \langle{ z_i, z_j \rangle} \\
& \hspace{.3cm} -\kappa_1 ( \langle{ z_j, z_c \rangle} -\langle{ z_c, z_j \rangle} ) \langle{ z_i, z_j \rangle} -\kappa_1 ( \langle{ z_c, z_i \rangle} -\langle{ z_i, z_c \rangle} ) \langle{ z_i, z_j \rangle} \\
& \hspace{.3cm} +m \| \dot{z}_j \|^2 \langle{ z_i, z_j \rangle} +m \| \dot{z}_i \|^2 \langle{ z_i, z_j \rangle} -2m  \langle{ \dot{z}_i, \dot{z}_j \rangle} \\
& = -\gamma \dot{g}_{ij} -\kappa_0 ( \langle{ z_i, z_c \rangle} +\langle{ z_c, z_j \rangle} )g_{ij} \\
& \hspace{0.3cm} -\kappa_1 ( \langle{ z_j, z_c \rangle} -\langle{ z_c, z_j \rangle} ) \langle{ z_i, z_j \rangle} -\kappa_1 ( \langle{ z_c, z_i \rangle} -\langle{ z_i, z_c \rangle} ) \langle{ z_i, z_j \rangle} \\
& \hspace{.3cm} -m ( \| \dot{z}_i \|^2 +\| \dot{z}_j \|^2 ) g_{ij} +m ( \| \dot{z}_j \|^2 -2  \langle{ \dot{z}_i, \dot{z}_j \rangle} +\| \dot{z}_i \|^2 ).
\end{align*}
Note that
\begin{align*}
( \langle{ z_i, z_c \rangle} +\langle{ z_c, z_j \rangle} )g_{ij} = \frac1N \sum_{k=1}^N( 2 -g_{ik} -g_{kj} )g_{ij}.
\end{align*}
This yields
\begin{align}
\begin{aligned} \label{F-20}
m \ddot{g}_{ij} & = -\gamma \dot{g}_{ij} -2\kappa_0 g_{ij} +\frac{\kappa_0}{N} \sum_{k=1}^N ( g_{ik} +g_{kj} ) g_{ij} \\
& \hspace{0.3cm} -\kappa_1 ( \langle{ z_j, z_c \rangle} -\langle{ z_c, z_j \rangle} ) \langle{ z_i, z_j \rangle} -\kappa_1 ( \langle{ z_c, z_i \rangle} -\langle{ z_i, z_c \rangle} ) \langle{ z_i, z_j \rangle} \\
& \hspace{.3cm} -m ( \| \dot{z}_i \|^2 +\| \dot{z}_j \|^2 ) g_{ij} +m ( \| \dot{z}_j \|^2 -2  \langle{ \dot{z}_i, \dot{z}_j \rangle} +\| \dot{z}_i \|^2 ).
\end{aligned}
\end{align}
We multiply \eqref{F-20} by $g_{ji}$ to find
\begin{align}\label{F-21}
\begin{aligned}
& m \ddot{g}_{ij}g_{ji} +\gamma \dot{g}_{ij}g_{ji} +\left[ 2\kappa_0 +m ( \| \dot{z}_i \|^2 +\| \dot{z}_j \|^2 ) \right] | g_{ij} |^2 \\
& \hspace{0.5cm} = \frac{\kappa_0}{N} \sum_{k=1}^N ( g_{ik} +g_{kj} ) | g_{ij} |^2 -\kappa_1 ( \langle{ z_j, z_c \rangle} -\langle{ z_c, z_j \rangle} ) \langle{ z_i, z_j \rangle} g_{ji} \\
& \hspace{0.7cm} -\kappa_1 ( \langle{ z_c, z_i \rangle} -\langle{ z_i, z_c \rangle} ) \langle{ z_i, z_j \rangle} g_{ji} +m ( \| \dot{z}_j \|^2 -2  \langle{ \dot{z}_i, \dot{z}_j \rangle} +\| \dot{z}_i \|^2 ) g_{ji}.
\end{aligned}
\end{align}
We sum up \eqref{F-21} over all $i, j$ and its complex conjugate to obtain the desired estimate.
\end{proof}

Next, we quote some useful Lemmas on the second-order Gronwall type differential inequality from \cite{C-H-Y} and \cite{H-K} without proofs.

\begin{lemma}\label{L6.2}
\cite{C-H-Y} Let $y = y(t)$ be a nonnegative $\mathcal C^2-$function satisfying the following differential inequality:
\begin{align*}
a\ddot{y} +b\dot{y} +cy +d \leq 0, \quad t > 0,
\end{align*}
where $a, b$ and $c$ are positive constants. Then, we have the following assertions:
\begin{enumerate}
\item
Suppose that $b^2 -4ac > 0$. Then, one has
\begin{align*}
y(t) & \leq -\frac{d}{c} +\bigg( y(0) +\frac{d}{c} \bigg) e^{-\nu_1 t} \\
& \quad +\frac{a}{\sqrt{b^2 -4ac}} \bigg( \dot{y}(0) +\nu_1 y(0) +\frac{2d}{b-\sqrt{b^2 -4ac}} \bigg) \big( e^{-\nu_2 t} -e^{-\nu_1 t} \big)
\end{align*}
where $\nu_1$ and $\nu_2$ are given as follows:
\begin{align*}
\nu_1 := \frac{b+\sqrt{b^2-4ac}}{2a} \quad \mbox{and} \quad \nu_2 := \frac{b-\sqrt{b^2-4ac}}{2a}.
\end{align*}
Moreover, if the following conditions hold:
\begin{align}\label{E-7}
y(0) +\frac{d}{c} < 0 \quad \mbox{and} \quad y'(0) +\nu_1y(0) +\frac{2d}{b-\sqrt{b^2 -4ac}} < 0,
\end{align}
then, $y(t)$ is uniformly bounded:
\begin{align*}
y(t) < -\frac{d}{c}.
\end{align*}
\item
Suppose that
\[ b^2 -4ac < 0. \]
Then, one has
\begin{align*}
y(t) \leq -\frac{4ad}{b^2} +e^{-\frac{b}{2a} t} \bigg[ y(0) +\frac{4ad}{b^2} +\bigg( \frac{b}{2a}y(0) +\dot{y}(0) +\frac{2d}{b} \bigg) t \bigg].
\end{align*}
Moreover, if the following conditions hold:
\begin{align}\label{E-1-3}
y(0) < -\frac{4ad}{b^2} \quad \mbox{and} \quad \frac{b}{2a}y(0) +\dot{y}(0) +\frac{2d}{b} < 0,
\end{align}
then, $y(t)$ is uniformly bounded:
\begin{align*}
y(t) < -\frac{4ad}{b^2}.
\end{align*}
\end{enumerate}
\end{lemma}

\begin{lemma}\label{L5.4}
\cite{H-K} Let $y = y(t)$ be a nonnegative $\mathcal C^2$-function satisfying the second-order differential inequality:
\begin{align*}
a\ddot{y} +b\dot{y} +cy \leq f, \quad t>0,
\end{align*}
where $a, b, c$ and $d$ are positive constants and $f = f(t)$ is a nonnegative $\mathcal C^1$-function which converges to zero as $t\to\infty$. Then, $y$ vanishes asymptotically:
\begin{align*}
\lim_{t\to\infty} y(t) = 0.
\end{align*}
\end{lemma}
Now, we are ready to provide a proof of our first main result on the complete aggregation of \eqref{F-0}. \newline

\noindent {\bf Proof of Theorem \ref{T4.1}}: Suppose that the sufficient frameworks $(\mathcal F_A1)$-$(\mathcal F_A2)$ or $(\mathcal F_B1)$-$(\mathcal F_B2)$ hold. Moreover, assume that initial data and natural frequency satisfy
\[
\|z_j^{in}\| = 1, \quad \Omega_j = 0, \quad j=1,\cdots,N.
\]
Let $\{ z_j \}$ be the global solution of \eqref{A-1}. Then, we claim:
\[ \lim_{t \to \infty} {\mathcal G}(t) = 0. \]
It follows from \eqref{F-16} that ${\mathcal G}$ satisfies
\begin{align*}
\begin{aligned}
m \ddot{{\mathcal G}} +\gamma \dot{{\mathcal G}} +4\kappa_0 {\mathcal G} & \leq \frac{2\kappa_0}{N^3} \sum_{i,j,k=1}^N ( |g_{ik}| +|g_{jk}| ) | g_{ij} |^2 +\frac{2\kappa_1}{N} \sum_{i=1}^N | \langle{ z_i, z_c \rangle} -\langle{ z_c, z_i \rangle} |^2 \\
& +\frac{2m}{N^2} \sum_{i,j=1}^N |\dot{g}_{ij}|^2 +\frac{2m}{N^2} \sum_{i,j=1}^N ( \| \dot{z}_j \|^2 +2  |\langle{ \dot{z}_i, \dot{z}_j \rangle}| +\| \dot{z}_i \|^2 ) \\
&=: \mathcal I_{11} +\mathcal I_{12} +\mathcal I_{13} +\mathcal I_{14},
\end{aligned}
\end{align*}
where
\begin{align*}
& \mathcal I_{11} := \frac{2\kappa_0}{N^3} \sum_{i,j,k=1}^N ( |g_{ik}| +|g_{jk}| ) | g_{ij} |^2, \quad \cI_{12} := \frac{2\kappa_1}{N} \sum_{i=1}^N | \langle{ z_i, z_c \rangle} -\langle{ z_c, z_i \rangle} |^2, \\
& \cI_{13} := \frac{2m}{N^2} \sum_{i,j=1}^N |\dot{g}_{ij}|^2, \quad \cI_{14} := \frac{2m}{N^2} \sum_{i,j=1}^N ( \| \dot{z}_j \|^2 +2  |\langle{ \dot{z}_i, \dot{z}_j \rangle}| +\| \dot{z}_i \|^2 ).
\end{align*}
Below, we provide estimates for ${\mathcal I}_{1i}$ one by one. \newline

\noindent$\bullet$~Case A (Estimate of $\cI_{11}$): We use the Cauchy-Swartz inequality to obtain
\begin{align*}
\begin{aligned}
&  \frac{2\kappa_0}{N^3} \sum_{i,j,k=1}^N |g_{ik}| \cdot | g_{ij} |^2 \\
& \hspace{1cm} \leq 2\kappa_0 \sqrt{N} \Bigg( \frac{1}{N^\frac32} \sum_{k=1}^N \sqrt{ \sum_{i=1}^N |g_{ik}|^2 } \Bigg) \cdot \Bigg( \frac{1}{N^2} \sum_{j=1}^N \sqrt{ \sum_{i=1}^N |g_{ij}|^4 } \Bigg) \\
& \hspace{1cm} \leq 2\kappa_0 \sqrt{N} \Bigg( \frac{1}{N^\frac12}  \sqrt{ \frac{1}{N}\sum_{k=1}^N \sum_{i=1}^N |g_{ik}|^2 } \Bigg) \cdot \Bigg( \frac{1}{N^2} \sum_{j=1}^N \sum_{i=1}^N |g_{ij}|^2 \Bigg) \\
& \hspace{1cm} \leq 2 \kappa_0 \sqrt{N} {\mathcal G}^\frac32.
\end{aligned}
\end{align*}
This implies
\begin{equation} \label{F-22}
\cI_{11} \leq 4 \kappa_0 \sqrt{N} {\mathcal G}^\frac32.
\end{equation}

\noindent$\bullet$~Case B (Estimate of $\cI_{12}$): By \eqref{F-22}, we have
\begin{equation} \label{F-23}
\cI_{12} \leq 2\kappa_1 {\mathcal R}_2(Z).
\end{equation}

\noindent$\bullet$~Case C (Estimate of $\cI_{13}$): We use
\[
|\dot g_{ij}|^2 = |\langle \dot z_i, z_j \rangle +\langle z_i, \dot z_j \rangle|^2 \leq 2(\| \dot z_i \|^2 +\| \dot z_j \|^2) \leq 4{\mathcal R}_1(\dot Z)
\]
to find
\begin{equation} \label{F-24}
\mathcal I_{13} \leq 8m {\mathcal R}_1(\dot Z).
\end{equation}

\noindent$\bullet$ Case D (Estimate of $\cI_{14}$): Similarly, we use
\[
\| \dot{z}_j \|^2 +2  |\langle{ \dot{z}_i, \dot{z}_j \rangle}| +\| \dot{z}_i \|^2 \leq 2(\| \dot z_i \|^2 +\| \dot z_j \|^2) \leq 4{\mathcal R}_1(\dot Z),
\]
to find
\begin{equation} \label{F-25}
\mathcal I_{14} \leq 8m {\mathcal R}_1(\dot Z).
\end{equation}
We combine all the estimates \eqref{F-22}, \eqref{F-23}, \eqref{F-24}, \eqref{F-25}  of $\cI_{1k}$'s to obtain
\begin{equation}\label{F-26}
m \ddot{{\mathcal G}} +\gamma \dot{{\mathcal G}} +4\kappa_0 {\mathcal G}  \leq 4\kappa_0 \sqrt{N} {\mathcal G}^\frac{3}{2} +2\kappa_1 {\mathcal R}_2(Z) +16m {\mathcal R}_1(\dot Z).
\end{equation}
Now, we derive a uniform bound for ${\mathcal G}$ using \eqref{F-26}:
\[
\sup_{0 \leq t < \infty} {\mathcal G}(t) < \frac{(1 -\delta)^2}{N}.
\]
We define a temporal set $\cT$ for $\delta \in (0,1)$:
\begin{align*}
\mathcal{T} := \{ T \in (0, \infty):  {\mathcal G}(t) < (1-\delta)^2/N, \quad \forall~t \in (0, T) \}.
\end{align*}
By initial conditions, the set $\mathcal{T}$ is nonempty. Hence we can define
\begin{align*}
T_* := \sup \mathcal{T}.
\end{align*}
Now we claim:
\[ T_* = \infty. \]
Suppose not, i.e.,
\[ T_* < \infty. \]
Then, we have
\begin{align}\label{E-1-1}
\lim_{t \to T_*-} {\mathcal G}(t) = \frac{(1 -\delta)^2}{N}.
\end{align}
On the other hand, it follows from \eqref{F-26} that, for $t \in (0, T_*)$, we have
\begin{align*}
\begin{aligned}
m \ddot{{\mathcal G}} +\gamma \dot{{\mathcal G}} +4\kappa_0 \delta {\mathcal G}  \leq 2\kappa_1 {\mathcal R}_2(Z) +16m {\mathcal R}_1(\dot Z).
\end{aligned}
\end{align*}
We use Corollary \ref{C5.1} to obtain
\begin{align*}
2\kappa_1 {\mathcal R}_2(Z) +16m {\mathcal R}_1(\dot Z) \leq 8\kappa_1 +16m M_1^2.
\end{align*}
Hence, one has
\begin{equation*}
m \ddot{{\mathcal G}} +\gamma \dot{{\mathcal G}} +4\kappa_0 \delta {\mathcal G}  \leq 8\kappa_1 +16m M_1^2, \quad t \in (0, T_*).
\end{equation*}
Note that $(\mathcal{F}_A1)$ and $(\mathcal{F}_B1)$ are the first and second case of Lemma \ref{L6.2}, respectively. Moreover, $(\mathcal{F}_A2)$ and $(\mathcal{F}_B2)$ satisfy the condition \eqref{E-7} and \eqref{E-1-3} of Lemma \ref{L6.2}, respectively. So, we apply Lemma \ref{L6.2} to obtain
\begin{align*}
& (\mathcal F_A) \implies {\mathcal G}(t) < \frac{8\kappa_1 +16mM_1^2}{4\kappa_0\delta} < \frac{(1-\delta)^2}{N}, \quad t \in (0, T_*), \\
& (\mathcal F_B) \implies \mathcal G(t) < \frac{4m}{\gamma^2}(8\kappa_1 +16mM_1^2) < \frac{(1-\delta)^2}{N}, \quad t \in (0, T_*),
\end{align*}
which contradicts to \eqref{E-1-1}. Therefore, we have $T_* = \infty$ and so that
\begin{align*}
m \ddot{{\mathcal G}} +\gamma \dot{{\mathcal G}} +4\kappa_0 \delta {\mathcal G}  \leq 2\kappa_1 {\mathcal R}_2(Z) +16m {\mathcal R}_1(\dot Z), \quad \forall ~t > 0.
\end{align*}
We use Theorem \ref{T4.1} and \eqref{F-13} to see 
\[ \lim_{t\to\infty} \big[ 2\kappa_1 {\mathcal R}_2(Z) +16m {\mathcal R}_1(\dot Z) \big] = 0. \]
Then, we can apply Lemma \ref{L5.4} to conclude that complete synchronization occurs. \qed
%%%%%%%%%%%%%%%%%%%%%%%%%%%%%%%%%%%%%%%%
%
%
%
%%%%%%%%%%%%%%%%%%%%%%%%%%%%%%%%%%%%%%%%%
\section{Emergence of practical aggregation} \label{sec:6}
\setcounter{equation}{0} 
In this section, we study the emergent dynamics of \eqref{A-1} with distinct set of natural frequency matrices $\Omega_j$'s. Unlike to the homogeneous ensemble in previous section, we cannot expect the emergence of complete aggregation in which all the states collapse to the same state. Instead, we study a weaker concept of aggregation estimate, namely practical aggregation introduced in Definition \ref{D1.1}. The key ingredient as in homogeneous ensemble is to derive a suitable second-order differential inequality for ${\mathcal G}$. In fact,  similar to \eqref{F-0-1}, we derive 
\[
m\ddot{{\mathcal G}} +\gamma \dot{{\mathcal G}} +4\kappa_0 \delta {\mathcal G} \leq 4\Omega^\infty +8\kappa_1 +\frac{16m}{\gamma^2}\big[ \Omega^\infty +2(\kappa_0 +\kappa_1) \big]^2, \quad \forall ~t > 0.
\]
Then, via the second-order Gronwall's lemma (Lemma \ref{L6.2}) and a suitable ansatz for $m = \frac{m_0}{\kappa_0^{1+\eta}}$, one can show
\[ {\mathcal G}(t) \lesssim \max \Big \{ \frac{1}{\kappa_0},~\frac{1}{\kappa_0^{\eta}} \Big \}, \quad \mbox{for $t \gg 1$}. \]
This clearly implies the desired practical aggregation estimate. 
\subsection{Derivation of Gronwall's inequality for ${\mathcal G}$} \label{sec:6.1}
Recall that the LHS model on the unit hermitian sphere $\|z_j \| = 1$:
\begin{equation}
\begin{cases} \label{E-0}
\displaystyle m \ddot{z}_j  = \frac{m}{\gamma} \Omega_j {\dot z}_j + \frac{m}{\gamma} \Omega_j v_j -\gamma \dot{z}_j + \Omega_j z_j  + \kappa_0(z_c-\langle{z_c, z_j}\rangle z_j) \\
\displaystyle \hspace{1cm} +\kappa_1(\langle{z_j, z_c}\rangle-\langle{z_c, z_j}\rangle)z_j - m \| v_j \|^2 z_j,   \\
\displaystyle (z_j, {\dot z}_j ) \Big|_{t  = 0+} = (z_j^{in}, w_j^{in}), \quad \langle z_j^{in}, w_j^{in} \rangle + \langle w_j^{in}, z_j^{in} \rangle - \frac{2}{\gamma} \Omega_j \langle z_j^{in}, z_j^{in} \rangle  = 0,
\end{cases}
\end{equation}
where $v_j = {\dot z}_j - \frac{1}{\gamma} \Omega_j z_j$.  \newline

Parallel to Lemma \ref{L5.3}, in the following lemma, we derive a dynamical system of $g_{ij}$.
\begin{lemma} \label{L6.1}
Let $\{ z_j \}$ be a global solution of \eqref{E-0}. Then, $|g_{ij} |^2$ satisfies
\begin{align}
\begin{aligned} \label{E-1}
& m \frac{d^2}{dt^2} | g_{ij} |^2 +\gamma \frac{d}{dt} | g_{ij} |^2 +2\left[ 2\kappa_0 +m ( \| v_i \|^2 +\| v_j \|^2 ) \right] | g_{ij} |^2 \\
&= -\frac{m}{\gamma} \big( \langle{ z_i, \Omega_j v_j \rangle} +\langle{ \Omega_i v_i, z_j \rangle} +\langle{ z_i, \Omega_j \dot{z}_j \rangle} +\langle{ \Omega_i \dot{z}_i, z_j \rangle} \big) g_{ji} -( \langle{ z_i, \Omega_j z_j \rangle} +\langle{ \Omega_i z_i, z_j \rangle} ) g_{ji} \\
&-\frac{m}{\gamma} \big( \langle{ z_j, \Omega_i v_i \rangle} +\langle{ \Omega_j v_j, z_i \rangle} +\langle{ z_j, \Omega_i \dot{z}_i \rangle} +\langle{ \Omega_j \dot{z}_j, z_i \rangle} \big) g_{ij} -( \langle{ z_j, \Omega_i z_i \rangle} +\langle{ \Omega_j z_j, z_i \rangle} ) g_{ij} \\
&+\frac{\kappa_0}{N} \sum_{k=1}^N ( g_{ik} +g_{ki} +g_{kj} +g_{jk} ) | g_{ij} |^2 +\kappa_1 ( \langle{ z_c, z_j \rangle} -\langle{ z_j, z_c \rangle} ) ( \langle{ z_i, z_j \rangle} -\langle{ z_j, z_i \rangle}) \\
&+\kappa_1 ( \langle{ z_i, z_c \rangle} -\langle{ z_c, z_i \rangle}) ( \langle{ z_i, z_j \rangle} -\langle{ z_j, z_i \rangle}) +2m\dot{g}_{ij}\dot{g}_{ji} \\
& -2m ( \langle{ \dot{z}_i, \dot{z}_j \rangle} g_{ji} +\langle{ \dot{z}_j, \dot{z}_i \rangle} g_{ij} ) +m ( \| v_i \|^2 +\| v_j \|^2 ) ( g_{ji} +g_{ij} ).
\end{aligned}
\end{align}
\end{lemma}
\begin{proof}
We use \eqref{E-0} to obtain
\begin{align*}
\langle{ z_i, m \ddot{z}_j \rangle} & = \frac{m}{\gamma} \langle{ z_i, \Omega_j v_j \rangle} +\frac{m}{\gamma} \langle{ z_i, \Omega_j \dot{z}_j \rangle} -\gamma \langle{z_i, \dot{z}_j \rangle} +\langle{ z_i, \Omega_j z_j \rangle} \\
& \hspace{.2cm} +\kappa_0 \langle{ z_i, z_c \rangle} -\kappa_0 \langle{ z_c, z_j \rangle} \langle{ z_i, z_j \rangle} +\kappa_1 ( \langle{ z_j, z_c \rangle} -\langle{ z_c, z_j \rangle} ) \langle{ z_i, z_j \rangle} -m \| v_j \|^2 \langle{ z_i, z_j \rangle}.
\end{align*}
Then, we have
\begin{align}
\begin{aligned} \label{E-2}
m \ddot{g}_{ij} &= -\langle{ z_i, m\ddot{z}_j \rangle} -2m \langle{ \dot{z}_i, \dot{z}_j \rangle} -\langle{ m\ddot{z}_i, z_j \rangle} \\
& = -\frac{m}{\gamma} \big( \langle{ z_i, \Omega_j v_j \rangle} +\langle{ \Omega_i v_i, z_j \rangle} +\langle{ z_i, \Omega_j \dot{z}_j \rangle} +\langle{ \Omega_i \dot{z}_i, z_j \rangle} \big) \\
& \hspace{0.2cm} -\langle{ z_i, \Omega_j z_j \rangle} -\langle{ \Omega_i z_i, z_j \rangle} -\gamma \dot{g}_{ij} -\kappa_0 \langle{ z_i, z_c \rangle} g_{ij} -\kappa_0 \langle{ z_c, z_j \rangle} g_{ij} \\
& \hspace{0.2cm} -\kappa_1 ( \langle{ z_j, z_c \rangle} -\langle{ z_c, z_j \rangle} ) \langle{ z_i, z_j \rangle} -\kappa_1 ( \langle{ z_c, z_i \rangle} -\langle{ z_i, z_c \rangle} ) \langle{ z_i, z_j \rangle} \\
& \hspace{.2cm} -2m \langle{ \dot{z}_i, \dot{z}_j \rangle} +m ( \| v_i \|^2 +\| v_j \|^2 ) -m ( \| v_i \|^2 +\| v_j \|^2 ) g_{ij}.
\end{aligned}
\end{align}
We substitute
\begin{align*}
\langle{ z_i, z_c \rangle} = \frac{1}{N} \sum_{k=1}^N (1 -g_{ik}) = 1 -\frac{1}{N} \sum_{k=1}^N h_{ik},
\end{align*}
into \eqref{E-2} and multiply both sides of \eqref{E-2} by $h_{ji}$ to get
\begin{align}
\begin{aligned} \label{E-3}
& m \ddot{g}_{ij} g_{ji} +\gamma \dot{g}_{ij} g_{ji} +\big[ 2\kappa_0 +m ( \| v_i \|^2 +\| v_j \|^2 ) \big] | g_{ij} |^2 \\
& \hspace{0.2cm} = -\frac{m}{\gamma} \big( \langle{ z_i, \Omega_j v_j \rangle} +\langle{ \Omega_i v_i, z_j \rangle} +\langle{ z_i, \Omega_j \dot{z}_j \rangle} +\langle{ \Omega_i \dot{z}_i, z_j \rangle} \big) g_{ji} +\frac{\kappa_0}{N} \sum_{k=1}^N ( g_{ik} +g_{kj} ) | g_{ij} |^2 \\
& \hspace{.5cm} -( \langle{ z_i, \Omega_j z_j \rangle} +\langle{ \Omega_i z_i, z_j \rangle} ) g_{ji} -\kappa_1 ( \langle{ z_j, z_c \rangle} -\langle{ z_c, z_j \rangle} ) \langle{ z_i, z_j \rangle} g_{ji} \\
& \hspace{.5cm} -\kappa_1 ( \langle{ z_c, z_i \rangle} -\langle{ z_i, z_c \rangle} ) \langle{ z_i, z_j \rangle} g_{ji} -2m \langle{ \dot{z}_i, \dot{z}_j \rangle} g_{ji} +m ( \| v_i \|^2 +\| v_j \|^2 ) g_{ji}.
\end{aligned}
\end{align}
We sum \eqref{E-3} and its complex conjugate to obtain the desired result.
\end{proof}

\vspace{0.2cm}

As in Corollary \ref{C5.1}, we observe the uniform bound of $\| v_j \|$. Since
\begin{align*}
\langle{ v_j, m \dot{v}_j \rangle} & = \frac{m}{\gamma}\langle{ v_j, \Omega_j v_j \rangle} -\gamma \| v_j \|^2 +\kappa_0 \langle{ v_j, z_c \rangle} -\kappa_0 \langle{ z_c, z_j \rangle} \langle{ v_j, z_j \rangle} \\
& \hspace{0.2cm} +\kappa_1 ( \langle{ z_j, z_c \rangle} -\langle{ z_c, z_j \rangle} ) \langle{ v_j, z_j \rangle} -m \| v_j \|^2 \langle{ v_j, z_j \rangle},
\end{align*}
we have
\begin{align}\label{E-4}
m\frac{d}{dt} \| v_j \|^2 \leq -2\gamma\| v_j \|^2 +4( \kappa_0 +\kappa_1 ) \| v_j \|,
\end{align}
where we use $\langle{ z_j, v_j \rangle} +\langle{ v_j, z_j \rangle} = 0$. \eqref{E-4} implies
\begin{align*}
\| v_j(t) \| \leq \bigg( \| v_j^{in} \| -\frac{2(\kappa_0 +\kappa_1)}{\gamma} \bigg) e^{-\frac{\gamma}{m} t} +\frac{2}{\gamma}(\kappa_0 +\kappa_1).
\end{align*}
Hence, we can obtain the uniform bound of $\| v_j \|$:
\begin{align}\label{E-5}
\| v_j \| \leq \max \bigg\{ \| v_1^{in} \|, \cdots, \| v_N^{in} \| \frac{2}{\gamma}(\kappa_0 +\kappa_1) \bigg\}, \quad j = 1,\cdots,N.
\end{align}
Also, we have the uniform bound of $\| \dot{z}_j \|$:
\begin{align}\label{E-6}
\| \dot{z}_j \| \leq \| v_j \| +\frac{\| \Omega_j \|_F}{\gamma} \leq \max \bigg\{\| v_1^{in} \|, \cdots, \| v_N^{in} \|, \frac{2}{\gamma}(\kappa_0 +\kappa_1) \bigg\} +\frac{\Omega}{\gamma}^\infty, \quad j = 1,\cdots,N,
\end{align}
where $\Omega^\infty := \max_{j} \| \Omega_j \|_F$.

\vspace{0.5cm}

\subsection{Proof of Theorem \ref{T4.2}} \label{sec:6.2} In this subsection, we provide a proof of our second main result on the emergence of practical aggregation. First, we begin with the derivation of a uniform bound for ${\mathcal G}$.  \newline

\noindent $\bullet$~Step A (Derivation of uniform bound for ${\mathcal G})$:  Suppose the framework $(\mathcal{F}_C1)$-$(\mathcal{F}_C2)$ hold, and let $Z = (z_1, \cdots, z_N)$ be a solution of \eqref{E-0}. Then, one has
\[
\sup_{0 \leq t < \infty} {\mathcal G}(t) < \frac{(1 -\delta)^2}{N}.
\]
It follows from \eqref{E-1} that
\begin{align}\label{E-8}
\begin{aligned}
m\ddot{{\mathcal G}} +\gamma \dot{{\mathcal G}} +4\kappa_0 {\mathcal G} &\leq 4\kappa_0 \sqrt{N} {\mathcal G}^\frac{3}{2} +\frac{4m\Omega^\infty}{\gamma} \Big[ {\mathcal R}_3(V) +\sqrt{{\mathcal R}_1(\dot{Z})} \Big] \\
&+12m {\mathcal R}_1(\dot{z}) +4m {\mathcal R}_3(V)^2 +4\Omega^\infty +8\kappa_1.
\end{aligned}
\end{align}
We define a temporal set $\cT$ for $\delta \in (0,1)$:
\begin{align*}
\mathcal{T} := \{ T \in (0, \infty):  {\mathcal G}(t) < (1-\delta)^2/N, \quad \forall~t \in (0, T) \}.
\end{align*}
By initial conditions, the set $\mathcal{T}$ is nonempty. Hence we can define
\begin{align*}
T_* := \sup \mathcal{T}.
\end{align*}
Now we claim:
\[ T_* = \infty. \]
Suppose not, i.e.,
\[ T_* < \infty. \]
Then, we have
\begin{align}\label{E-9}
\lim_{t \to T_*-} {\mathcal G}(t) = \frac{(1 -\delta)^2}{N}.
\end{align}
On the other hand, it follows from \eqref{E-8} that, for $t \in (0, T_*)$, we have
\begin{align}\label{E-10}
\begin{aligned}
m\ddot{{\mathcal G}} +\gamma \dot{{\mathcal G}} +4\kappa_0 \delta {\mathcal G} &\leq \frac{4m\Omega^\infty}{\gamma} \Big[ {\mathcal R}_3(V) +\sqrt{{\mathcal R}_1(\dot{z})} \Big]  \\&+12m {\mathcal R}_1(\dot{z}) +4m {\mathcal R}_3(V)^2 +4\Omega^\infty +8\kappa_1.
\end{aligned}
\end{align}
We use \eqref{E-5}, \eqref{E-6} and $(\mathcal{F}_C2)_1$ to obtain
\begin{align*}
\begin{aligned}
& \frac{4m\Omega^\infty}{\gamma} {\mathcal R}_3(V) +4m {\mathcal R}_3(V)^2 \\
& \hspace{.2cm} \leq \frac{8m\Omega^\infty(\kappa_0 +\kappa_1)}{\gamma^2} +\frac{16m (\kappa_0 +\kappa_1)^2}{\gamma^2} = \frac{8m}{\gamma^2} (\kappa_0 +\kappa_1)\big[ \Omega^\infty +2(\kappa_0 +\kappa_1) \big],
\end{aligned}
\end{align*}
and 
\begin{align*}
\begin{aligned}
&\Omega^{\infty}{\gamma}\sqrt{{\mathcal R}_1(\dot{z})} +12m {\mathcal R}_1(\dot{z}) \\
& \hspace{0.2cm} \leq \frac{4m\Omega^\infty \big[ \Omega^\infty +2(\kappa_0 +\kappa_1) \big]}{\gamma^2} +\frac{12m \big[ \Omega^\infty +2(\kappa_0 +\kappa_1) \big]^2}{\gamma^2} \\
&\hspace{0.2cm} = \frac{4m}{\gamma^2} \big[ \Omega^\infty +2(\kappa_0 +\kappa_1) \big] \big[ 4\Omega^\infty +6(\kappa_0 +\kappa_1) \big].
\end{aligned}
\end{align*}
Hence, one has
\begin{equation} \label{E-11}
\frac{4m\Omega^\infty}{\gamma} \Big[ {\mathcal R}_3(V) +\sqrt{{\mathcal R}_1(\dot{z})} \Big] +12m {\mathcal R}_1(\dot{z}) +4m {\mathcal R}_3(V)^2 \leq \frac{16m}{\gamma^2} \big[ \Omega^\infty +2(\kappa_0 +\kappa_1) \big]^2. 
\end{equation}
Now, we combine \eqref{E-10} and \eqref{E-11} to get 
\[
m\ddot{{\mathcal G}} +\gamma \dot{{\mathcal G}} +4\kappa_0 \delta {\mathcal G} \leq 4\Omega^\infty +8\kappa_1 +\frac{16m}{\gamma^2}\big[ \Omega^\infty +2(\kappa_0 +\kappa_1) \big]^2, \quad t \in (0, T_*).
\]
Note that $(\mathcal{F}_C1)$ is the first case of Lemma \ref{L6.2} and $(\mathcal{F}_C2)_2$ and $(\mathcal{F}_C2)_3$ satisfy the condition \eqref{E-7} of Lemma \ref{L6.2}. So, we apply Lemma \ref{L6.2} to obtain
\begin{equation}\label{E-12}
{\mathcal G}(t) < \frac{1}{4\kappa_0\delta} \Bigg( 4\Omega^\infty +8\kappa_1 +\frac{16m}{\gamma^2}\big[ \Omega^\infty +2(\kappa_0 +\kappa_1) \big]^2 \Bigg) < \frac{(1-\delta)^2}{N}, \quad t \in (0, T_*),
\end{equation}
which contradicts to \eqref{E-9}. Therefore, we have $T_* = \infty$.

\vspace{.2cm}

Now, we are ready to provide a proof of Theorem \ref{T4.2}.

\vspace{.2cm}

\noindent $\bullet$~Step B (Derivation of practical aggregation estimate):~It follows from \eqref{E-12} that
\begin{equation} \label{E-12-1}
{\mathcal G}(t) < \frac{\Omega^\infty +2\kappa_1}{\kappa_0\delta} +\frac{4m \big[ \Omega^\infty +2(\kappa_0 +\kappa_1) \big]^2}{\gamma^2\kappa_0\delta}, \quad \forall ~t>0.
\end{equation}
For the case $\kappa_0 \geq \max \big\{ \Omega^\infty, 2\kappa_1 \big\}$, we have
\begin{align*}
{\mathcal G}(t) < \frac{\Omega^\infty +2\kappa_1}{\kappa_0\delta} +\frac{4m \kappa_0 \Big[ \frac{\Omega^\infty}{\kappa_0} +2 +\frac{2\kappa_1}{\kappa_0} \Big]^2}{\gamma^2\delta} \leq \frac{\Omega^\infty +2\kappa_1}{\kappa_0\delta} +\frac{64}{\gamma^2\delta} m \kappa_0, \quad \forall ~t>0.
\end{align*}
Hence as 
\begin{equation} \label{E-14}
 m\kappa_0 \rightarrow 0 \quad \mbox{and} \quad \kappa_0 \rightarrow \infty, 
 \end{equation}
one has a practical synchronization. To satisfy the constraints \eqref{E-14},  we assume that there exist $m_0 > 0$ and $\eta > 0$ such that
\begin{equation}\label{E-15}
m = \frac{m_0}{\kappa_0^{1+\eta}}.
\end{equation}
Then, it follows from \eqref{E-12-1} and \eqref{E-15} that, for $\kappa_0 \geq \max \big\{ \Omega^\infty, 2\kappa_1 \big\}$,
\begin{align*}
{\mathcal G}(t) < \frac{\Omega^\infty +2\kappa_1}{\kappa_0\delta} +\frac{64}{\gamma^2\delta} \cdot \frac{m_0}{\kappa_0^\eta}, \quad \forall ~t>0,
\end{align*}
which implies the desired result. $\qed$

\begin{remark}
There could be a question about the possibility for second inequality of $(\mathcal{F}_C2)_2$. We verify it holds for sufficiently large $\kappa_0$. We substitute \eqref{E-15} into the second inequality of $(\mathcal{F}_C2)_2$ to obtain
\begin{align}\label{E-16}
\frac{\Omega^\infty +2\kappa_1}{\kappa_0\delta} +\frac{4m_0 \big[ \Omega^\infty +2(\kappa_0 +\kappa_1) \big]^2}{\delta\gamma^2\kappa_0^{2+\eta}} < \frac{(1-\delta)^2}{N}.
\end{align}
One can observe that left hand side of \eqref{E-16} converges to zeros as $\kappa_0$ goes to infinity.
\end{remark}

\section{Conclusion} \label{sec:7}
\setcounter{equation}{0}
In this paper, we have studied emergent behaviors of the second-order LHS model which can be realized as a second-order extension of the first-order LHS model introduced in authors' earlier work \cite{H-P0, H-P1, H-P2}. For a homogeneous ensemble with the same natural frequency matrix $\Omega_j = \Omega$, we provided emergence of complete aggregation in the sense that all states aggregate to the same state asymptotically. For this, under a suitable set of system parameters and initial data with a finite energy, we show that the two-point correlation functions  between states tend to zero asymptotically, which denote the formation of complete aggregation. By linear stability analysis, we also showed that the incoherent state and bi-polar state are linearly unstable. In contrast, for a heterogeneous ensemble, 
we provided a sufficient framework leading to practical aggregation which means that state diameter can be made small by increasing the principle coupling strength. Of course, there are several issues to be discussed in a future work. For example, we only considered positive coupling strengths (attractive couplings) in this work. However, the coupling strengths can be negative, i.e., repulsive couplings or they can be time-dependent or state-dependent which make asymptotic dynamics more richer. We leave these interesting issues in a future work. 

%%%%%%%%%%%%%%%%%%%%%%%%%%%%%%%%%%%%%%%%%%%%%%%%%%%%%%%%

\end{document}